\documentclass[journal,twoside,web]{ieeecolor}

\usepackage{generic}
\usepackage{cite}
\usepackage{amsmath,amssymb,amsfonts}
 \usepackage{algorithm}
\usepackage{algorithmicx}
\usepackage{algpseudocode}
\usepackage{graphicx}
\usepackage{subcaption}
\usepackage{textcomp}

\newtheorem{theorem}{\bf Theorem}

\newtheorem{remark}{\bf Remark}

\newtheorem{proposition}{\bf Proposition}

\def\QED{~\rule[-1pt]{5pt}{5pt}\par\medskip}

\newcommand{\bx}{{\bf x}}
\newcommand{\by}{{\bf y}}
\newcommand{\bz}{{\bf z}}
\newcommand{\bw}{{\bf w}}
\newcommand{\bv}{{\bf v}}

\begin{document}
\title{Continuous-Time Channel Gain Control for Minimum-Information Kalman-Bucy Filtering}
\author{Takashi Tanaka, Vrushabh Zinage, Valery Ugrinovskii, and Mikael Skoglund 
\thanks{T. Tanaka and V. Zinage were supported by NSF Award 1944318. V. Ugrinovskii was supported by the Australian Research Council under 
  Discovery Projects funding scheme (project DP200102945).
M. Skoglund was funded in part by
 the Swedish research council under contract 2019-03606.}
\thanks{T. Tanaka and V. Zinage are with the Department of Aerospace Engineering and Engineering Mechanics at the University of Texas at Austin, Austin, TX 78712 USA (e-mail: ttanaka@utexas.edu, vrushabh.zinage@utexas.edu). }
\thanks{V. Ugrinovskii is with the School of Engineering and Information Technology at the University of New South Wales Canberra, Canberra ACT, Australia (e-mail: v.ougrinovski@adfa.edu.au). }
\thanks{M. Skoglund is with the Division of Information Science and
  Engineering at the KTH Royal Institute of Technology, Stockholm,
  Sweden (e-mail: skoglund@kth.se).}}

\maketitle

\begin{abstract}
We consider the problem of estimating a continuous-time Gauss-Markov source process observed through a vector Gaussian channel with an adjustable channel gain matrix. For a given (generally time-varying) channel gain matrix, we provide formulas to compute (i) the mean-square estimation error attainable by the classical Kalman-Bucy filter, and (ii) the mutual information between the source process and its Kalman-Bucy estimate.
We then formulate a novel ``optimal channel gain control problem'' where the objective is to control the channel gain matrix strategically to minimize the weighted sum of these two performance metrics. To develop insights into the optimal solution, we first consider the problem of controlling a time-varying channel gain over a finite time interval. A necessary optimality condition is derived based on Pontryagin's minimum principle. For a scalar system, we show that the optimal channel gain is a piece-wise constant signal with at most two switches. We also consider the problem of designing the optimal time-invariant gain to minimize the average cost over an infinite time horizon. A novel semidefinite programming (SDP) heuristic is proposed and the exactness of the solution is discussed.
\end{abstract}

\begin{IEEEkeywords}
Kalman filters, Information theory, Continuous time systems, Networked control systems,  Optimal control. 
\end{IEEEkeywords}

\section{Introduction}
\label{sec:introduction}

In this paper, we consider the problem of estimating a continuous-time Gauss-Markov process from a noisy observation through a vector Gaussian channel with an adjustable channel gain matrix. Once the channel gain is given, the optimal causal estimate in the minimum mean-square (MSE) sense is readily computable by the celebrated Kalman-Bucy filter. 
However, our focus in this paper is on a generalized problem in which the channel gain matrix must also be designed strategically. 
Specifically, our objective is to design a (generally time-varying) channel gain matrix to minimize the weighted sum of (i) the MSE attainable by the resulting Kalman-Bucy filter, and (ii) the mutual information between the source process and its Kalman-Bucy estimate.
Since choosing a ``larger'' channel gain results in a smaller MSE and a larger mutual information, these two performance criteria are in a trade-off relationship in general. Therefore, controlling the channel gain to attain the sweet spot is a nontrivial problem. 

The problem we study is motivated by the causal source coding scenario in which an observed time series must be encoded, compressed, and transmitted over a digital communication media to a remote decoder who tries to reproduce the original signal without delay (e.g., \cite{fine1964properties, witsenhausen1979structure,TatikondaThesis, teneketzis2006structure, yuksel2013stochastic, guo2021optimal} and references therein).
Since the rate-distortion trade-off under such a causality constraint sets a fundamental performance limitation for control systems over resource-constrained communication networks, the optimal codec design has been investigated extensively in the networked control systems literature \cite{matveev2004problem,nair2007feedback,silva2010framework,yuksel2013jointly,tanaka2017lqg,kostina2019rate}. 
An approach based on the information-theoretic causal rate-distortion function (also called non-anticipative, zero-delay, or sequential rate-distortion function) \cite{gorbunov1974prognostic, tatikonda2004stochastic,derpich2012improved,charalambous2013nonanticipative,tanaka2016semidefinite} attempts to estimate the minimum bit-rate required for remote estimation using the mutual information between the source process and the reproduced process.  
For stationary Gaussian sources, \cite{derpich2012improved} showed that the minimum bit-rate required to reproduce the signal within a given MSE distortion criterion is lower bounded by the causal rate-distortion function and that the conservatism of this lower bound is less than a constant space-filling gap.
Subsequent works by the authors of \cite{tanaka2016semidefinite} and \cite{stavrou2018optimal} developed algorithms to compute the causal rate-distortion function for discrete-time Gauss-Markov sources with the MSE distortion criteria. The work \cite{tanaka2016semidefinite} proved the ``channel-filter separation principle,'' asserting that the optimal test channel (a stochastic kernel on the reproduced signal given the source signal that minimizes mutual information subject to a given MSE distortion constraint) can always be realized by a memoryless Gaussian channel with an appropriately chosen channel gain followed by a Kalman filter. This result implies that, for discrete-time Gauss-Markov processes, computing the optimal channel gain that minimizes the mutual information under a given MSE distortion constraint reveals a fundamental limitation of real-time data compression for remote estimation. Although the existing works on causal rate-distortion theory are largely limited to discrete-time settings (exceptions include \cite{guo2021optimal}), the results we present in this paper allow us to solve the continuous-time counterpart of the same problem, and will contribute to the development of causal source codes for continuous-time processes.

The trade-off between the mutual information (I) and the minimum mean-square error (MMSE) is known as the \emph{I-MMSE relationship} in the information theory literature. For random variable observed through Gaussian channels, Guo et al. \cite{guo2005mutual} showed that the derivative of the mutual information with respect to the channel SNR (signal-to-noise ratio) is equal to half the MMSE. They also considered random processes observed through Gaussian channels and provided a simple connection between causal and non-causal MMSE. The I-MMSE relationship for continuous-time source processes was derived by a predating work by Duncan \cite{duncan1970calculation}.
Kadota et al. \cite{kadota1971mutual} considered the problem of estimating continuous-time source over Gaussian channel with feedback (the source is causally affected by channel output).
Weissman et al. \cite{weissman2013directed} further studied the cases with feedback and derived a fundamental relationship between directed information and MMSE.
Palomar and Verd\'u \cite{palomar2006gradient} extended the result by \cite{guo2005mutual} to vector Gaussian channels and derived an explicit formula to compute the gradient of mutual information with respect to channel parameters. They applied this result to obtain a gradient ascent algorithm for channel precoder design to maximize the input-output mutual information subject to input power constraints.

\subsection{Contribution of this paper}
Contributions of this paper are summarized as follows:
\begin{enumerate}
\item We derive a new mutual information formula (Theorem~\ref{theo:MI}) that generalizes both the main result of Duncan \cite{duncan1970calculation} and our earlier result \cite[Theorem 1]{tanaka2017optimal}.
\item We show the existence of a measurable solution to the optimal channel gain control problem over a finite time interval based on Filippov's result \cite{optimal_control_daniel_liberzon}. We also derive a necessary optimality condition based on Pontryagin's minimum principle.
\item We prove that the optimal channel gain control over a finite time interval for a scalar source process is piecewise constant with at most two discontinuities.
\item We also consider the problem of finding the optimal time-invariant channel gain for minimum-information Kalman-Bucy filtering. We propose a semidefinite programming (SDP) relaxation to compute an optimal solution candidate. We present the result of extensive numerical experiments suggesting that the relaxation is in fact exact, although we are currently not aware of a theoretical guarantee on the exactness. 
\end{enumerate}
Preliminary versions of the results appeared in the authors' conference publications \cite{tanaka2017optimal} and \cite{zinage2021optimal}.
The items 2) and 3) extended the results in \cite{zinage2021optimal} to accommodate more general channel gain constraints. 
The item 4) did not appear in any of prior publications.

\subsection{Notation}
We will use notation $x_{[t_1, t_2]}=\{x_t: t_1 \leq t \leq t_2\}$ to denote a continuous-time signal. Bold symbols like $\bx$ will be used to denote random variables. We assume all the random variables considered in this paper are defined on the probability space $(\Omega, \mathcal{F}, \mathcal{P})$. The probability distribution $p_\bx$ of an $(\mathcal{X}, \mathcal{A})$-valued random variable $\bx$ is defined by
\[
p_\bx(A)=\mathcal{P}\{\omega\in\Omega : \bx(\omega)\in A\}, \; \forall A \in \mathcal{A}.
\]
If $\bx_1$ and $\bx_2$ are both $(\mathcal{X}, \mathcal{A})$-valued random variables, the relative entropy from $\bx_2$ to $\bx_1$ is defined by
\[
D(p_{\bx_1} \| p_{\bx_2})=\int \log \frac{dp_{\bx_1}}{dp_{\bx_2}} dp_{\bx_1}
\]
provided that the Radon-Nikodym derivative $\frac{dp_{\bx_1}}{dp_{\bx_2}}$ exists, and $+\infty$ otherwise. The mutual information between two random variables $\bx$ and $\by$ is defined by
\[
I(\bx;\by)=D(p_{\bx\by} \| p_\bx \otimes p_\by)
\]
where $p_{\bx\by}$ and $p_\bx \otimes p_\by$ denote the joint and product distributions, respectively.

\section{Problem formulation}
\subsection{System description}
Let $(\Omega,\mathcal{F}, \mathcal{P})$ be a complete probability space and suppose $\mathcal{F}_t\subset \mathcal{F}$ form a non-decreasing family of $\sigma$-algebras.
Suppose $(\bw_t, \mathcal{F}_t)$ and $(\bv_t, \mathcal{F}_t)$ are $\mathbb{R}^n$-valued, mutually independent standard Wiener processes with respect to $\mathcal{P}$.
Define the random process to be estimated as an $n$-dimensional Ito process
\begin{equation}
\label{eq:source}
    d\bx_t=A\bx_t dt + B d\bw_t, \;\;\; t\in [t_0, t_1]
\end{equation}
with $\bx_{t_0}\sim\mathcal{N}(0, X_0)$, where $X_0\succeq 0$ is a given covariance matrix and $[t_0, t_1]$ is a time interval. We assume $A$ and $B$ are Hurwitz and nonsingular matrices, respectively. Let $C_t:[t_0, t_1] \rightarrow \mathbb{R}^{n\times n}$ be a measurable function representing the time-varying channel gain.
Setting $\bz_t \triangleq C_t \bx_t$, the channel output is an $n$-dimensional signal
\begin{equation}
\label{eq:measurement}
    d\by_t=\bz_t dt + d\bv_t, \;\;\; t\in [t_0, t_1]
\end{equation}
with $\by_{t_0}=0$. Based on the channel output $\by_t$, the causal MMSE estimate $\hat{\bx}_t\triangleq \mathbb{E}(\bx_t|\mathcal{F}_t^{\by})$ is computed, where $\mathcal{F}_t^{\by} \subset \mathcal{F}$ denotes the  $\sigma$-algebra generated by $\by_{[0,t]}$.
The causal MMSE estimate can be computed by the 
Kalman-Bucy filter
\begin{equation}
\label{eqkbfilter}
d\hat{\bx}_t=A\hat{\bx}_tdt+X_tC_t^\top (d\by_t-C_t\hat{\bx}_tdt), \;\;\; t\in [t_0, t_1]
\end{equation}
with $\hat{\bx}_{t_0}=0$. Here, $X_t$ is the unique solution to the matrix Riccati differential equation
\begin{equation}
\label{eqriccati}
\dot{X_t}=AX_t+X_tA^\top-X_tC_t^\top C_tX_t+BB^\top, \;\; t\in [t_0, t_1]
\end{equation}
with the initial condition $X_{t_0}=X_0\succeq 0$. The system architecture considered in this paper is shown in Fig.~\ref{fig:channel}.

\begin{figure}[h]
\centering
\includegraphics[width=8.5cm]{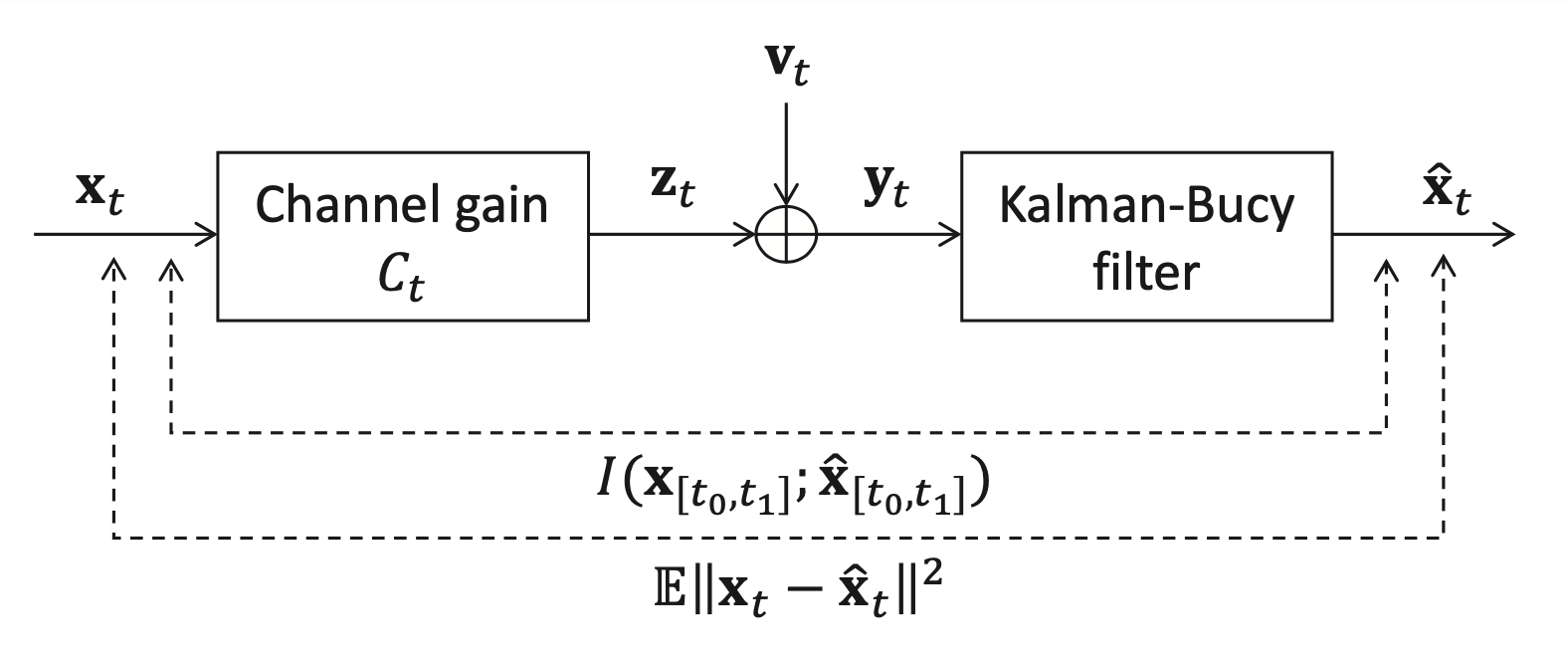}
\caption{System architecture and performance criteria.}
\label{fig:channel}
\end{figure}

\subsection{Performance criteria}
In this paper, we consider the problem of optimally controlling the time-varying channel gain $C_t$. The optimality is characterized in terms of the mean-square error and the mutual information, as shown in Fig.~\ref{fig:channel}.
\subsubsection{Mean-square error (MSE)}
The first criterion is the MMSE achieved by the Kalman-Bucy filter.
\begin{equation}
\label{eq:mmse}
 \int_{t_0}^{t_1} \mathbb{E} \| \bx_t-\hat{\bx}_t\|^2 dt =\int_{t_0}^{t_1} \mathrm{Tr} (X_t) \; dt.   
\end{equation}
\subsubsection{Mutual information}
The second performance criterion is the mutual information $I(\bx_{[t_0, t_1]}; \hat{\bx}_{[t_0, t_1]})$. The next theorem
extends the result of Duncan \cite{duncan1970calculation} and provides a key formula to compute $I(\bx_{[t_0, t_1]}; \hat{\bx}_{[t_0, t_1]})$ explicitly. 
\begin{theorem}
\label{theo:MI}
\normalfont Let the random processes $\bx_{[t_0,t_1]}$ and $\hat{\bx}_{[t_0,t_1]}$ be defined as above. Then
\[
I(\bx_{[t_0, t_1]}; \hat{\bx}_{[t_0, t_1]})=\frac{1}{2}\int_{t_0}^{t_1} \mathbb{E}\|C_t (\bx_t-\hat{\bx}_t)\|^2 dt.
\]
\end{theorem}
\begin{proof}
The following identity is shown in \cite{duncan1970calculation}:
\begin{equation}
\label{eq:thm1_1}
I(\by_{[t_0, t_1]}; \bz_{[t_0, t_1]})=\frac{1}{2}\int_{t_0}^{t_1} \mathbb{E}\|C_t (\bx_t-\hat{\bx}_t)\|^2 dt.
\end{equation}
Due to the property of the Kalman-Bucy filter, $\hat{\bx}_{[t_0, t_1]}$ is a sufficient statistic of $\by_{[t_0, t_1]}$ for $\bx_{[t_0, t_1]}$. Thus,
\begin{equation}
\label{eq:thm1_2}
I(\bx_{[t_0, t_1]}; \by_{[t_0, t_1]})=I(\bx_{[t_0, t_1]}; \hat{\bx}_{[t_0, t_1]}).
\end{equation}
It is left to show
\begin{equation}
\label{eq:thm1_3}
I(\bx_{[t_0, t_1]}; \by_{[t_0, t_1]})=I(\by_{[t_0, t_1]}; \bz_{[t_0, t_1]}).
\end{equation}
This fact follows from the dependency structure of the processes $\bx_{[t_0, t_1]}, \by_{[t_0, t_1]}$ and $\bz_{[t_0, t_1]}$. Indeed, under the law $\mathcal{P}$, 
the conditional independence 
\[
p_{\bz_{[t_0, t_1]}|\bx_{[t_0, t_1]},\by_{[t_0, t_1]}}=p_{\bz_{[t_0, t_1]}|\bx_{[t_0, t_1]}}
\]
holds because $\bz_{[t_0, t_1]}$ is a function of $\bx_{[t_0, t_1]}$.
Furthermore, since $\bz_{[t_0, t_1]}$ is independent of $\bv_{[t_0, t_1]}$ under
$\mathcal{P}$, \eqref{eq:measurement} implies the conditional independence
\[
p_{\by_{[t_0, t_1]}|\bx_{[t_0, t_1]},\bz_{[t_0, t_1]}}=p_{\by_{[t_0, t_1]}|\bz_{[t_0, t_1]}}.
\]
These facts allows us to conclude that the data-processing inequality holds with equality \cite{CoverThomas} (see also \cite[Chapter 7]{gray2011entropy}), implying \eqref{eq:thm1_3}.
\end{proof}

Theorem~\ref{theo:MI} is more general than our previous result \cite[Theorem 1]{tanaka2017optimal}, which was only applicable to time-invariant channel gains. The proof has also been simplified significantly by an application of the data-processing inequality. To the best of our knowledge, this result has not appeared in the literature. 

Theorem~\ref{theo:MI} implies that the mutual information $I(\by_{[t_0, t_1]}; \bz_{[t_0, t_1]})$ can be expressed as
\begin{equation}
\label{eqn:mutual_information}
I(\bx_{[t_0, t_1]}; \hat{\bx}_{[t_0, t_1]})=\frac{1}{2}\int_{t_0}^{t_1} \text{Tr}(C_tX_t C_t^\top)dt.
\end{equation}
\subsection{Problem setup}
The central question we study in the paper is how to choose a measurable function $C_t:[t_0, t_1]\rightarrow \mathbb{R}^{n\times n}$ to minimize the weighted sum of  the MSE \eqref{eq:mmse} and the mutual information \eqref{eqn:mutual_information}. Notice that in general it is not possible to minimize these two quantities simultaneously. To see this, suppose we choose $C_t=kC \;\; \forall t$ where $k\geq 0$ is a scalar and $(A,C)$ is an observable pair. As $k\rightarrow +\infty$, the MSE tends to zero while the mutual information tends to $+\infty$.
Introducing a trade-off parameter $\alpha>0$,\footnote{The parameter $\alpha$ is the Lagrange multiplier in view of the hard-constrained version of the problem studied in \cite{tanaka2017optimal}.} the main problem we study in this paper is formulated as follows:
\begin{subequations}
\label{eqn:main_optimization_problem}
\begin{align}
   \underset{C_t}{\mathrm{min}}&\;\; \int_{t_0}^{t_1}\;\mathbb{E}\|\bx_t-\hat{\bx}_t\|^2dt+2\alpha I(\bx_{[t_0, t_1]}; \hat{\bx}_{[t_0, t_1]}) \label{eqn:main_optimization_problem1} \\
   \text{ s.t. }&\;\;\; C_t^\top C_t \preceq \gamma I \;\; \forall \;\;t\in [t_0, t_1]. \label{eqn:main_optimization_problem2}
\end{align}
\end{subequations}
The constraint \eqref{eqn:main_optimization_problem2} imposes an upper bound on the allowable channel gain.
Using \eqref{eq:mmse} and \eqref{eqn:mutual_information}, the problem \eqref{eqn:main_optimization_problem} can be written more explicitly as
\begin{subequations}
\begin{align}
     \underset{C_t}{\mathrm{min}}\quad&\int_{t_0}^{t_1}\mathrm{Tr}(X_t)dt+\alpha \int_{t_0}^{t_1}\mathrm{Tr}(C_tX_tC_t^\top)dt\\
    \mathrm{s.t.}\quad&\dot{X}_t=AX_t+X_tA^\top-X_tC_t^\top C_tX_t+BB^\top \\
    &X_{t_0}=X_{0}\\
    &C_t^\top C_t \preceq \gamma I \quad\forall\;\; t\in[t_0,t_1].
\end{align}
\label{problem:optimization_problem}
\end{subequations}
Introducing $U_t\triangleq C_t^\top C_t\succeq 0$, this can be written as an equivalent optimal control problem with state $X_t$ and control input $U_t$:
\begin{subequations}
\label{problem:optimization_problem_control}
\begin{align}
    \underset{U_t}{\mathrm{min}}\quad&\int_{t_0}^{t_1}\mathrm{Tr}(X_t+\alpha U_tX_t)dt  \label{problem:optimization_problem_control_a}\\
    \mathrm{s.t.}\quad&\dot{X}_t=AX_t+X_tA^\top-X_tU_tX_t+BB^\top\\
    &X_{t_0}=X_{0}\\
    &U_t\succeq 0,\;\;\; U_t\preceq \gamma I \quad\forall\;\; t\in[t_0,t_1].\label{eqn:admissible_controls}
\end{align}
 \label{problem:optimization_problem_control}
\end{subequations}
The minimization is over the space of measurable functions $U_t: [t_0, t_1]\rightarrow \mathbb{S}_+^n (=\{M\in \mathbb{R}^{n\times n}: M \succeq 0\})$.
\begin{remark}
\normalfont The equivalence between
\eqref{problem:optimization_problem} and
\eqref{problem:optimization_problem_control} implies that optimal
solutions to the main problem \eqref{eqn:main_optimization_problem},
if they exist, are not unique. Namely, if $U_t^\star$ is an optimal solution to  \eqref{problem:optimization_problem_control}, then both $\bar{C}_t$ and $\tilde{C}_t$ are optimal solutions to \eqref{problem:optimization_problem} if $U_t^\star=\bar{C}_t^\top \bar{C}_t = \tilde{C}_t^\top \tilde{C}_t$.
\end{remark}

\begin{remark}
\normalfont For simplicity, we assume that $C_t$ is a square matrix throughout this paper. 
A generalization to the case with $C_t\in\mathbb{R}^{m \times n}$ where $m \geq n$ is straightforward. However, a technical difficulty arises if $m<n$ is required. To see this, notice that $C_t\in\mathbb{R}^{m \times n}$ implies  
$\text{rank}(C_t^\top C_t)\leq m$. Therefore, an additional non-convex constraint $\text{rank}(U_t)\leq m \; \forall t\in[t_0, t_1]$ must be included in \eqref{problem:optimization_problem_control} to maintain the equivalence between \eqref{problem:optimization_problem} and 
\eqref{problem:optimization_problem_control}. This type of difficulty has been observed in the sensor design problems in the literature, as in \cite{bansal1989simultaneous}. 
\end{remark}

In this paper, we are also interested in the optimal time-invariant channel gain $C_t=C\in \mathbb{R}^{n\times n}$ that minimizes the average cost over a long time horizon:
\begin{subequations}
\label{problem:time_invariant}
\begin{align}
    \min_{C\in\mathbb{R}^{n\times n}}\quad& \limsup_{t_1\rightarrow +\infty} \frac{1}{t_1-t_0}  \bigg\{ \int_{t_0}^{t_1}\mathbb{E}\|\bx_t-\hat{\bx}_t\|^2 dt \nonumber \\
    &+2\alpha I(\bx_{[t_0, t_1]}; \hat{\bx}_{[t_0, t_1]}) \bigg\} \label{problem:time_invariant_a}\\
    \mathrm{s.t. }\qquad &X_{t_0}=X_{0}, \;\; C^\top C  \preceq \gamma I. \label{problem:time_invariant_b}
\end{align}
\end{subequations}
This can be written as an equivalent optimization problem:
\begin{subequations}
\label{problem:time_invariant_riccati}
\begin{align}
    \min_{C\in\mathbb{R}^{n\times n}}\quad& \limsup_{t_1\rightarrow +\infty} \frac{1}{t_1-t_0}  \int_{t_0}^{t_1}\text{Tr}(X_t+\alpha CX_t C^\top )dt \label{problem:time_invariant_riccati_a}\\
    \mathrm{s.t. }\qquad &\dot{X}_t=AX_t+X_tA^\top-X_tC^\top C X_t+BB^\top \label{problem:time_invariant_riccati_b} \\
    &X_{t_0}=X_{0}, \;\; C^\top C  \preceq \gamma I. \label{problem:time_invariant_riccati_c}
\end{align}
\end{subequations}

\section{Optimality condition}
In this section, we briefly revisit some general results in optimal control theory to discuss the existence of an optimal control for \eqref{problem:optimization_problem_control} and to derive a necessary optimality condition.
Consider the following Lagrange-type optimal control problem with a fixed end time and a free end point.
\begin{subequations}
\begin{align}
    &\underset{u_t}{\mathrm{min}}\quad \int_{t_0}^{t_1}L(x_\tau,u_\tau)d\tau\\
    &\mathrm{s.t.}\quad \dot{x}_t=f(x_t,u_t),\;\;x_t\in\mathbb{R}^n,\;u_t\in\mathcal{U}\label{eqn:dynamics}\\
    &\quad\;\;\;\;\;x_{t_0}=x_0.
    \label{problem:optimal_control_problem_initial_condition}
\end{align}
\label{problem:optimal_control_problem}
\end{subequations}
Suppose that an admissible control input is a measurable function $u_t:[t_0,t_1]\rightarrow \mathcal{U}$ where $\mathcal{U}\subset \mathbb{R}^m$ is a compact set. We assume that $f(x_t,u_t)$, $\frac{\partial f(x_t,u_t)}{\partial x_t}$, $L(x_t,u_t)$ and $\frac{\partial L(x_t,u_t)}{\partial x_t}$ are continuous on $\mathbb{R}^n\times \mathcal{U}\times (t_0,t_1)$. 
Notice that problem  \eqref{problem:optimization_problem_control} is in the form \eqref{problem:optimal_control_problem} if $\mathcal{U}$ is selected as $\mathcal{U}=\{U\in \mathbb{S}_+^n: U\preceq \gamma I \}$.

\subsection{Existence of optimal control}
\label{sec:existence}
We first invoke the following useful result by Filippov \cite{filippov1962certain}:
\begin{theorem}
\label{theo:filippov}
\normalfont \textbf{(Filippov's theorem): \cite{optimal_control_daniel_liberzon}}
Given a control system \eqref{eqn:dynamics} with $u_t\in\mathcal{U}$, assume that its solutions exist on a time interval $[t_0,t_1]$ for all controls and that for every $x_t$ the set $\{f(x_t,u_t):u_t\in\mathcal{U}\}$ is compact and convex. Then the reachable set $R_t(x_0)$ is compact for each $t\in[t_0,t_1]$.
\end{theorem}

Theorem~\ref{theo:filippov} is applicable to guarantee the existence of an optimal control for  \eqref{problem:optimization_problem_control}.
Specifically, one can convert the original Langrange-type problem \eqref{problem:optimization_problem_control} into an equivalent Mayer-type problem by introducing an auxiliary state $x_t^{\text{aux}}$ satisfying $x_{t_0}^{\text{aux}}=0$ and $\dot{x}_t^{\text{aux}}=\text{Tr}(X_t+\alpha U_t X_t)$. 
In the Mayer form, the original problem of minimizing \eqref{problem:optimization_problem_control_a} becomes the problem of minimizing $x_{t_1}^{\text{aux}}$ over the reachable set at $t=t_1$. Since the premises of Theorem~\ref{theo:filippov} are satisfied by the obtained Mayer-type problem, we can conclude that the reachable set at $t=t_1$ is compact. Therefore, Weierstrass' extreme value theorem guarantees the existence of an optimal solution. 

\subsection{Pontryagin Minimum Principle}
We next invoke a version of Pontryagin's Minimum Principle for the fixed-endtime free-endpoint optimal control problem \eqref{problem:optimal_control_problem}.
\begin{theorem}
\label{theo:pmp}
\normalfont [Theorem 5.10, \cite{athans2013optimal_book}] 
Suppose there exists an optimal solution to \eqref{problem:optimal_control_problem}. Let $u^\star_t:[t_0,t_1]\rightarrow \mathcal{U}$ be an optimal control input and $x^\star_t:[t_0,t_1]\rightarrow \mathbb{R}^n $ be the corresponding state trajectory. Then, there exists a function $p^\star_t:[t_0,t_1]\rightarrow \mathbb{R}^n$ such that the following conditions hold for the Hamiltonian $H$ defined as
\begin{align}
    H(x_t,p_t,u_t)=L(x_t,u_t)+p^\top_t f(x_t,u_t):
\end{align}
\begin{itemize}
\item[(i)] $x^\star_t$ and $p^\star_t$ satisfy the following canonical equations:
\begin{align}
    &\dot{x}^\star_t=\frac{\partial H(x^\star_t,p^\star_t,u^\star_t)}{\partial p_t},\quad\dot{p}^\star_t=-\frac{\partial H(x^\star_t,p^\star_t,u^\star_t)}{\partial x_t}\nonumber
\end{align}
with boundary conditions $x_{t_0}=x_0$ and $p_{t_1}=0$.
\item[(ii)] $\underset{u_t\in \mathcal{U}}{\mathrm{min}}\;\; H(x^\star_t,p^\star_t,u_t)=H(x^\star_t,p^\star_t,u^\star_t)$ for all $t\in[t_0,t_1]$.
\end{itemize}
\end{theorem}
For our problem \eqref{problem:optimization_problem_control}, the Hamiltonian is defined as 
\begin{align}
    H(X_t,P_t,U_t)=&\mathrm{Tr}(X_t+\alpha U_t X_t)+\nonumber\\
    &\left<P_t, AX_t+X_t A^\top-X_tU_tX_t+BB^\top \right>\nonumber\\
    =&\mathrm{Tr}(P_t(AX_t+X_t A^\top+BB^\top))+\mathrm{Tr}(X_t)\nonumber\\
    &+\mathrm{Tr}((\alpha X_t-X_tP_tX_t)U_t).\nonumber
\end{align}
Thus, the necessary optimality condition provided by Theorem~\ref{theo:pmp} is given by the canonical equations
\begin{subequations}
\begin{align}
    &\dot{X}_t=AX_t+X_tA^\top-X_tU_tX_t+BB^\top \\
    &\dot{P}_t=P_tX_tU_t+U_tX_tP_t\!-\!P_tA\!-\!A^\top P_t\!-\!I\!- \!\alpha U_t
\end{align}
\label{eqn:canonical_vector_case}
\end{subequations}
with boundary conditions $X_{t_0}=X_0$ and $P_{t_1}=0$, and
\begin{align}
    U_t^\star=\underset{U_t\in\mathcal{U}}{\mathrm{argmin}}\;\mathrm{Tr}[(\alpha X_t-X_tP_tX_t)U_t]
    \label{eqn:u_star_matrix}
\end{align}
where $\mathcal{U}=\{U\in\mathbb{S}_+^n: U\preceq \gamma I \}$.

\section{Optimal solution: Scalar case}
\label{sec:scalar}
In this section, we restrict our attention to a special case with scalar systems (i.e., $n=1$) to be able to solve the optimality condition \eqref{eqn:canonical_vector_case}
and \eqref{eqn:u_star_matrix} explicitly.
In what follows, we assume $A=a<0$ and $B=1$. 
The canonical equations \eqref{eqn:canonical_vector_case} are simplified as
\begin{subequations}
\label{eqn:region_2_governing_equations_xp}
\begin{align}
    &\dot{x}_t=2ax_t-x_t^2u_t+1\label{eqn:region_2_governing_equations_x}\\
    & \dot{p}_t=2x_tp_tu_t-2ap_t-1-\alpha u_t\label{eqn:region_2_governing_equations_p}
\end{align}
\label{eqn:canonical}
\end{subequations}
with $x_{t_0}=x_0$ and $p_{t_1}=0$. Due to the original meaning of $x_0$ as a covariance of the initial value of the underlying process $\mathbf{x}_0$, we assume that $x_0\geq 0$. Since \eqref{eqn:region_2_governing_equations_x} is a monotone system in the sense of \cite{angeli2003monotone}, we have $x_t\geq 0$ for all $t\in[t_0,t_1]$. The optimal control $u_t$ is given by
\begin{align}
    u_t^\star&=\underset{0\leq u_t\leq \gamma}{\mathrm{argmin}}\;\; x_t(\alpha-x_tp_t)u_t\nonumber\\
    &= \left\{ \begin{array}{l}
0\quad\text{if}\quad p_tx_t<\alpha
\\
\gamma \quad\text{if}\quad p_tx_t>\alpha
\\
u^\star\in[0,\gamma]\quad\text{if}\quad p_tx_t=\alpha
\end{array}
\right.
\label{eqn:u_star_argmin}
\end{align}
The main result of this section is summarized in the next theorem:
\begin{theorem}
\normalfont For any $x_0>0$, $\alpha>0$, $a<0$ and specified time interval $[t_0,t_1]$, an optimal control exists and satisfies \eqref{eqn:region_2_governing_equations_xp} and \eqref{eqn:u_star_argmin}.
If $0\leq 2a/\sqrt{\alpha}+1/\alpha \leq \gamma$, the optimal control is a piecewise constant function that can take three possible values $u_t=0, 2a/\sqrt{\alpha}+1/\alpha, \gamma$. Otherwise, the optimal control is a piecewise constant function that can take two possible values $u_t=0, \gamma$ (bang-bang control). 
In all cases, the optimal control has at most two discontinuities.
\label{theorem:main_theorem}
\end{theorem}
\begin{proof}
The existence of an optimal solution follows from the discussion in Section~\ref{sec:existence}.
The rest of the statement will be established in Sections~\ref{subsec:phase_potrait} and Section~\ref{sec:analytical_solution} below. Explicit expressions for the optimal control are also given in Section~\ref{sec:analytical_solution}.
\end{proof}

\subsection{Phase Portrait analysis\label{subsec:phase_potrait}}

\begin{figure}[t]
\centering
\includegraphics[width=7.5cm]{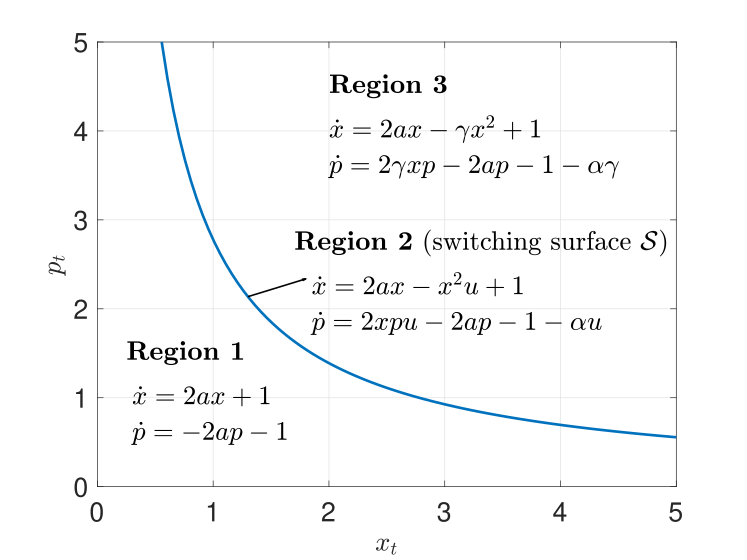}
\caption{Illustration of different regions in the phase space.}
\label{fig:regions}
\end{figure}
To analyze the canonical equations \eqref{eqn:canonical}, consider the vector field defined by the right-hand side of \eqref{eqn:canonical} and \eqref{eqn:u_star_argmin}. Because of the classification in \eqref{eqn:u_star_argmin}, the vector field is discontinuous on the switching surface $\mathcal{S}$ characterized by $p_tx_t=\alpha$. We divide the domain $\{(x,p)|x\geq0\}$ into Regions 1, 2 and 3 as shown in Fig.~\ref{fig:regions}. In Fig.~\ref{fig:regions}, we only show the positive orthant $\{(x,p)|x\geq0, p\geq 0\}$ since the region $\{(x,p)|x\geq0, p< 0\}$ plays no role in the following analysis.
Denote the vector field in Regions 1, 2 and 3 (see Fig. \ref{fig:cases}) as $f_1, f_2$ and $f_3$, respectively. From \eqref{eqn:canonical} and \eqref{eqn:u_star_argmin}, we have
\begin{align}
    f_1 :& \begin{cases} \dot{x}_t=2ax_t+1 \\
    \dot{p}_t=-2ap_t-1 \end{cases} \label{eqn:region1_xp_f1} \\
    f_2 :& \begin{cases}\dot{x}_t=2ax_t-x_t^2u_t+1 \\
    \dot{p}_t=2x_tp_tu_t-2ap_t-1-\alpha u_t \end{cases}  \label{eqn:region2_xp_f2} \\
    f_3 :& \begin{cases} \dot{x}_t=2ax_t-\gamma x_t^2+1\\
    \dot{p}_t=2\gamma x_tp_t-2ap_t-1-\alpha \gamma. \end{cases} \label{eqn:region3_xp_f3}
\end{align}
\subsubsection{Local solutions in Regions 1 and 3} 
The vector field in Region 1 is characterized by the linear differential equation \eqref{eqn:region1_xp_f1} whose general solution is given by
\begin{align}
    & x_t=\frac{k_1e^{2at}-1}{2a},\quad p_t=\frac{k_2e^{-2at}-1}{2a}
    \label{eqn:region1_xp}
\end{align}
where $k_1$ and $k_2$ are constants. 
On the other hand, the vector field in Region 3 characterized by \eqref{eqn:region3_xp_f3} is nonlinear. However,  \eqref{eqn:region3_xp_f3} belongs to the class of scalar Riccati differential equations and admits an analytical solution given as follows:
\begin{subequations}
\begin{align}
      &x_t=\frac{1}{\gamma} \left(a+c-\frac{2c}{k_3e^{2 ct}+1} \right) \label{eqn:region3_x}\\
      &    p_{t}= k_4\frac{\left(k_3e^{2 ct}+1\right)^2}{e^{2 ct}}+ \frac{(1+\alpha \gamma)\left(k_3e^{2 ct}+1\right)}{2ck_3e^{2 ct}}.\label{eqn:region3_p}
\end{align}
\label{eqn:region3_xp}
\end{subequations}
Here, $k_3$, $k_4$ are constants and $c=\sqrt{a^2+\gamma}$.
\subsubsection{Stationary points} 
The location of a stationary point in the phase portrait changes depending on the value of $\alpha$. Noticing that  $0<(a+\sqrt{a^2+\gamma})^2/\gamma^2<1/4a^2$ for all $a<0$ and $\gamma>0$, the following three cases can occur:
\begin{itemize}
    \item Case A: $1/4a^2<\alpha$. In this case, the phase portrait has a unique stationary point in Region 1 located at 
    \begin{align}
    E=(x_e,p_e)=(-1/2a,-1/2a). \label{eqn:equilibrium_region_1}
    \end{align}
    It is not possible for $f_2$ to have a stationary point in Region 2 no matter what value of $u_t\in[0,\gamma]$ is chosen. A stationary point cannot exist in Region 3 either.
    \item Case B: $(a+\sqrt{a^2+\gamma})^2/\gamma^2\leq\alpha\leq 1/4a^2$. In this case, a stationary point cannot exist in Region 1 or in Region 3. However, the point 
    \begin{align}
    E=(x_e,p_e)=(\sqrt{\alpha},\sqrt{\alpha})\label{eqn:equilibrium_region_2}
    \end{align}
    in Region 2 is a stationary point if $u_t$ is set to
    \begin{equation}
    \label{eq:u_star}
    u^\star=2a/\sqrt{\alpha}+1/\alpha.
    \end{equation}
    From the present assumption that $(a+\sqrt{a^2+\gamma})^2/\gamma^2 \leq\alpha\leq 1/4a^2$, it can be shown that the value of $u^\star$ in \eqref{eq:u_star} satisfies $0\leq u^\star \leq \gamma$. No other point in Region 2 can be a stationary point. 
    \item Case C: $\alpha<(a+\sqrt{a^2+\gamma})^2/\gamma^2$. In this case, the phase portrait has a unique stationary point 
    \begin{align}
    (x_e,p_e)=\left(\frac{a+\sqrt{a^2+\gamma}}{\gamma},\frac{1+\alpha\gamma}{2\sqrt{a^2+\gamma}}\right)\label{eqn:equilibrium_region_3}
    \end{align}
    in Region 3. No stationary point can exists in Regions 1 and 2.
\end{itemize}
The vector field in each case is depicted in Fig.~\ref{fig:cases}.

\begin{figure*}[ht]
 \captionsetup[subfigure]{justification=centering}
 \centering
 \begin{subfigure}{0.33\textwidth}
{\includegraphics[scale=0.16]{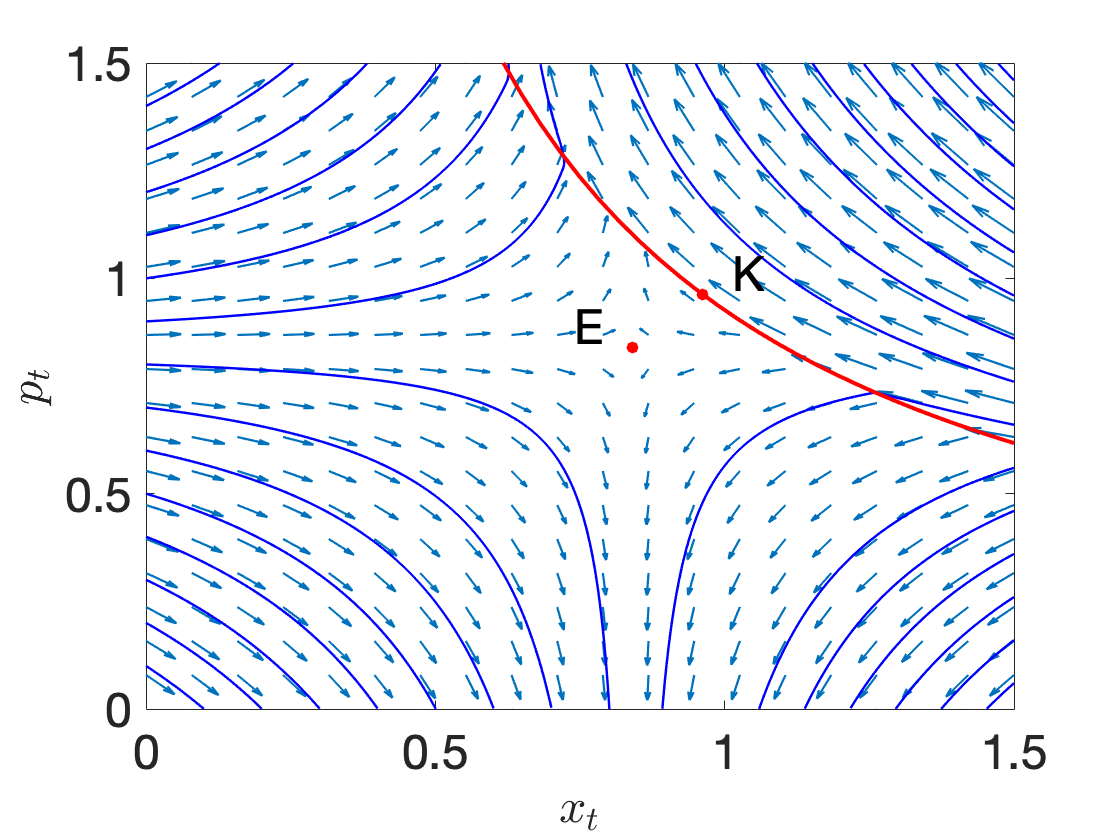}}
\caption{Case A ($\alpha=0.926$)}
\label{fig:casea}
 \end{subfigure}
\label{fig:approximation_error_x_position}
 \begin{subfigure}{0.33\textwidth}
{\includegraphics[scale=0.16]{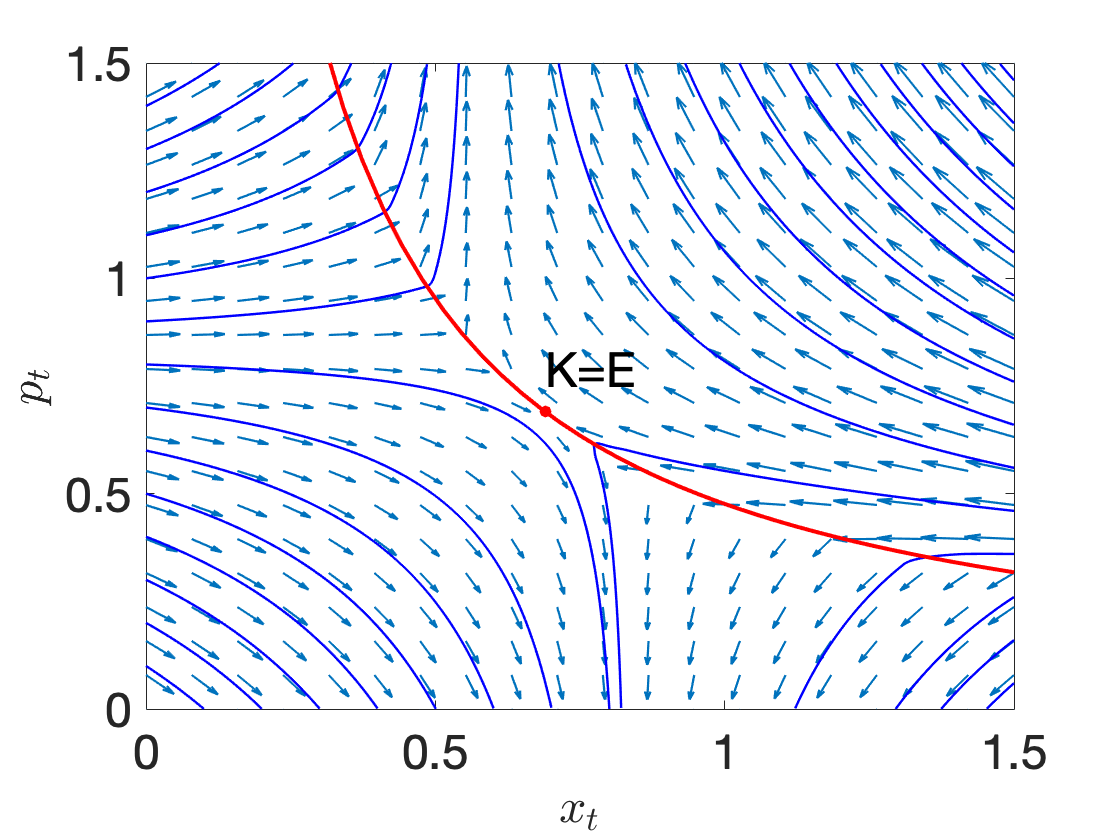}}
 \caption{Case B ($\alpha=0.476$)}
\label{fig:caseb}
 \end{subfigure}
\label{fig:approximation_error_v_velocity}
 \begin{subfigure}{0.32\textwidth}
{\includegraphics[scale=0.16]{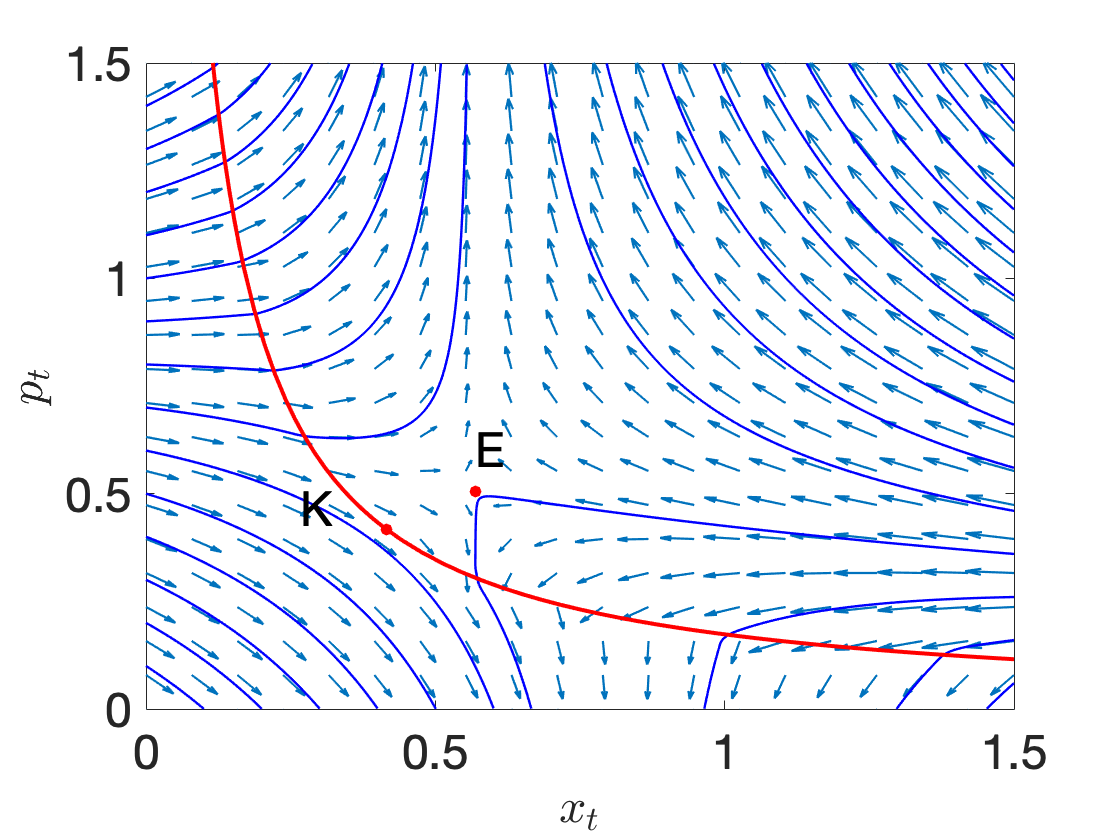}}
\caption{Case C ($\alpha=0.173$)}
\label{fig:casec}
\end{subfigure}
 \caption{Phase portraits for various cases ($a=-0.595,\;\gamma=1$).}
\label{fig:cases}
\end{figure*}

\subsubsection{Switching behavior}
To understand the behavior of the solution near the switching surface $\mathcal{S}$, we need to analyze the directions of $f_1, f_2$ and $f_3$ with respect to $\mathcal{S}$ in the neighborhood of $\mathcal{S}$. Noticing that $\mathcal{S}$ is a level set of the function $V(x,p)=xp$, this can be studied by checking the signs of the Lie derivatives $L_{f_1}V$, $L_{f_2}V$ and $L_{f_3}V$ evaluated on $\mathcal{S}$. Notice that
\begin{align}
    L_{f_1}V&=\frac{\partial V}{\partial x}\dot{x}+\frac{\partial V}{\partial p}\dot{p}=p\dot{x}+x\dot{p}\nonumber\\
    &=p(2ax+1)-x(2ap+1)=p-x\nonumber \\
    L_{f_2}V&=p\dot{x}+x\dot{p}=x^2pu+p-x-\alpha xu\nonumber \\      
L_{f_3}V&=p\dot{x}+x\dot{p}\nonumber\\
    &=p(2ax-\gamma x^2+1)+x(2\gamma xp-2ap-1-\alpha\gamma)\nonumber\\
        &=\gamma x^2p+p-x-\alpha\gamma x. \nonumber
\end{align}
Therefore, on the surface $\mathcal{S}$ (i.e., when $px=\alpha$), we have
\[
L_{f_1}V|_\mathcal{S}=L_{f_2}V|_\mathcal{S}=L_{f_3}V|_\mathcal{S}=\alpha/x-x.
\]
This indicates that all the vector fields $f_1$, $f_2$ and $f_3$ define a consistent direction with respect to $\mathcal{S}$ everywhere on $\mathcal{S}$. Namely, they cross $\mathcal{S}$ ``upward" in the portion where $x < \sqrt{\alpha}$, ``downward" where $x > \sqrt{\alpha}$, and are tangential to $\mathcal{S}$ at the point 
\[
K=(\sqrt{\alpha},\sqrt{\alpha}).
\]
The point $K$ becomes a stationary point when $(a+\sqrt{a^2+\gamma})^2/\gamma^2\leq\alpha\leq 1/4a^2$ (Case B) and $u_t$ is set to be \eqref{eq:u_star}. This coincides with $E$ defined by \eqref{eqn:equilibrium_region_2}. 

The analysis above has the important implication that the phase portraits in Fig.~\ref{fig:cases} are free from the ``chattering" solutions in all cases, and that the solution concept of Caratheodory \cite{discontinuous_dynamical_system_cortes} is sufficient to describe the solutions crossing the switching surface $\mathcal{S}$.  
However, special attention is needed to Case B where the uniqueness of the solution is lost. To see this, consider the family of the trajectories $(x_t, p_t)$ that ``stay" on $E=K$ for an arbitrary duration as follows:
\begin{itemize}
    \item $(x_t, p_t)$ solves \eqref{eqn:region1_xp_f1} or \eqref{eqn:region3_xp_f3} for $t_0\leq t \leq t'$ with  $(x_{t'}, p_{t'})=(\sqrt{\alpha}, \sqrt{\alpha})$;
    \item $(x_t, p_t)=(\sqrt{\alpha}, \sqrt{\alpha})$ for $t'\leq t \leq t''$;
    \item $(x_t, p_t)$ solves \eqref{eqn:region1_xp_f1} or \eqref{eqn:region3_xp_f3} for $t'' \leq t_1$ with  $(x_{t''}, p_{t''})=(\sqrt{\alpha}, \sqrt{\alpha})$.
\end{itemize}
It is easy to check that all these trajectories are Caratheodory solutions to the canonical equations regardless of the choice of $t'$ and $t''$. 

\subsection{Analytical solution \label{sec:analytical_solution}}
Using the phase portraits depicted in Fig.~\ref{fig:cases}, we now solve the boundary value problem \eqref{eqn:canonical} and \eqref{eqn:u_star_argmin} with the initial state condition $x_{t_0}=x_0 (\geq 0)$ and the terminal costate condition $p_{t_1}=0$. In what follows, the solution to this boundary value problem is simply referred to as the optimal solution. It is convenient to consider Cases A, B and C separately.

\subsubsection{Case A}
In this case, the initial coordinate of $(x_{t_0}^\star,p_{t_0}^\star)$ of the optimal solution is either in the green or the pink regions illustrated in Fig.~\ref{fig:case_a_color}. The boundaries of these regions are defined by the switching surface $\mathcal{S}$ and the separatrices converging to the point $E$.
\subsubsection*{\underline{Subcase A-1} ($(x_{t_0}^\star,p_{t_0}^\star)$ is in the green region)}
Consider the particular solution to the vector field $f_1$ satisfying the boundary conditions $x_{t_0}=x_0$ and $p_{t_1}=0$:
\begin{subequations}
\label{eqn:xp_case1}
    \begin{align}
   & \bar{x}_t=\frac{(2ax_0+1)e^{2a(t-t_0)}-1}{2a}\label{eqn:x_case1}\\
   & \bar{p}_t=\frac{e^{-2a(t-t_1)}-1}{2a}.\label{eqn:p_case1}
\end{align}
\end{subequations}
If we compute $(\bar{x}_{t_0}, \bar{p}_{t_0})$ from \eqref{eqn:xp_case1} and find it is in Region 1 (i.e., $\bar{x}_{t_0}\bar{p}_{t_0} \leq \alpha$), then Subcase A-1 applies.  
Since the optimal solution is entirely in Region 1 there is no switching.
\begin{figure*}[ht]
 \captionsetup[subfigure]{justification=centering}
 \centering
 \begin{subfigure}{0.33\textwidth}
{\includegraphics[width=6.2cm]{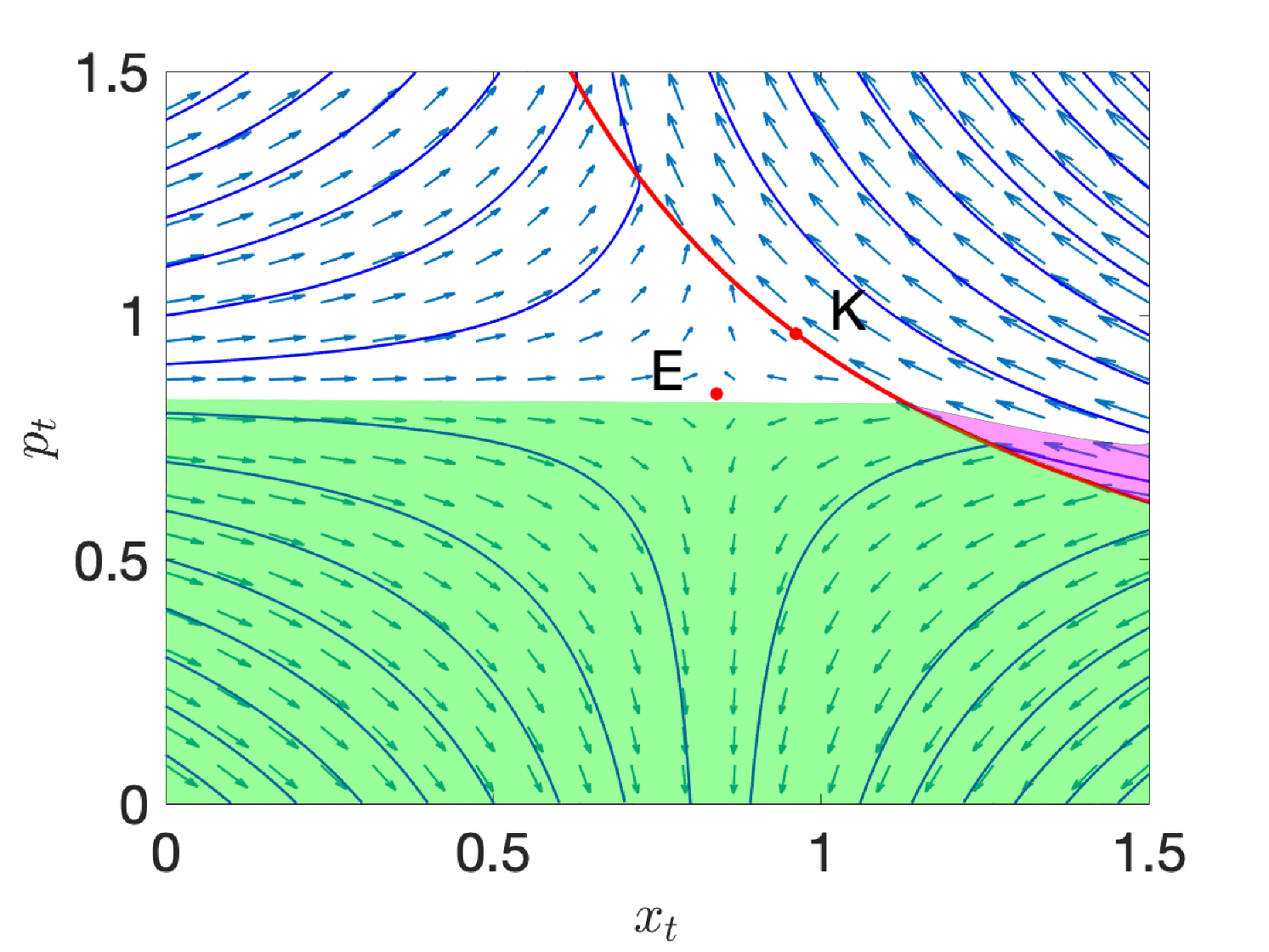}}
\caption{Case A.}
\label{fig:case_a_color}
 \end{subfigure}
 \begin{subfigure}{0.33\textwidth}
 \includegraphics[width=6.2cm]{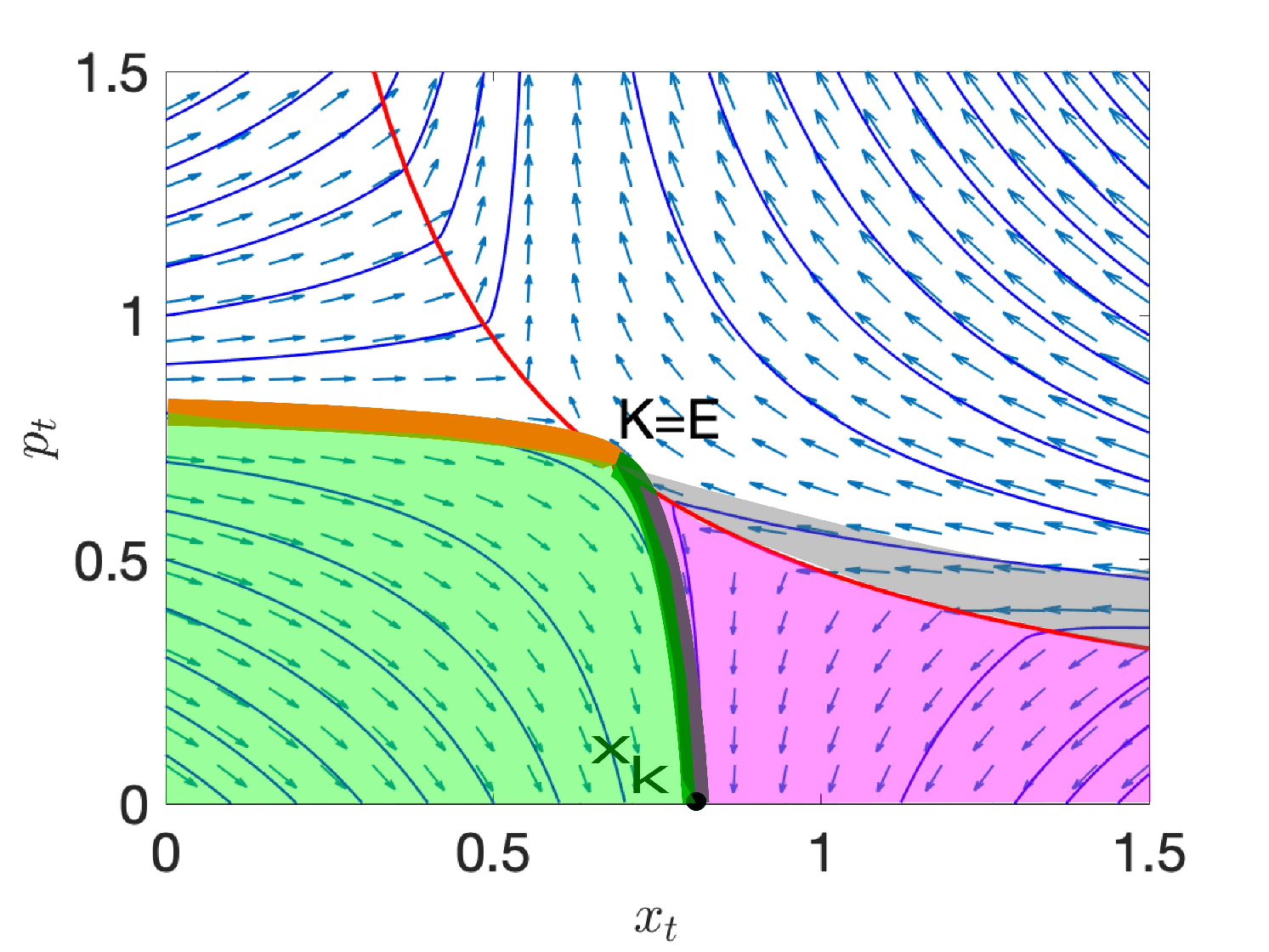}
 \caption{Case B.}
\label{fig:case_b_color}
 \end{subfigure}
 \begin{subfigure}{0.32\textwidth}
 \includegraphics[width=6.2cm]{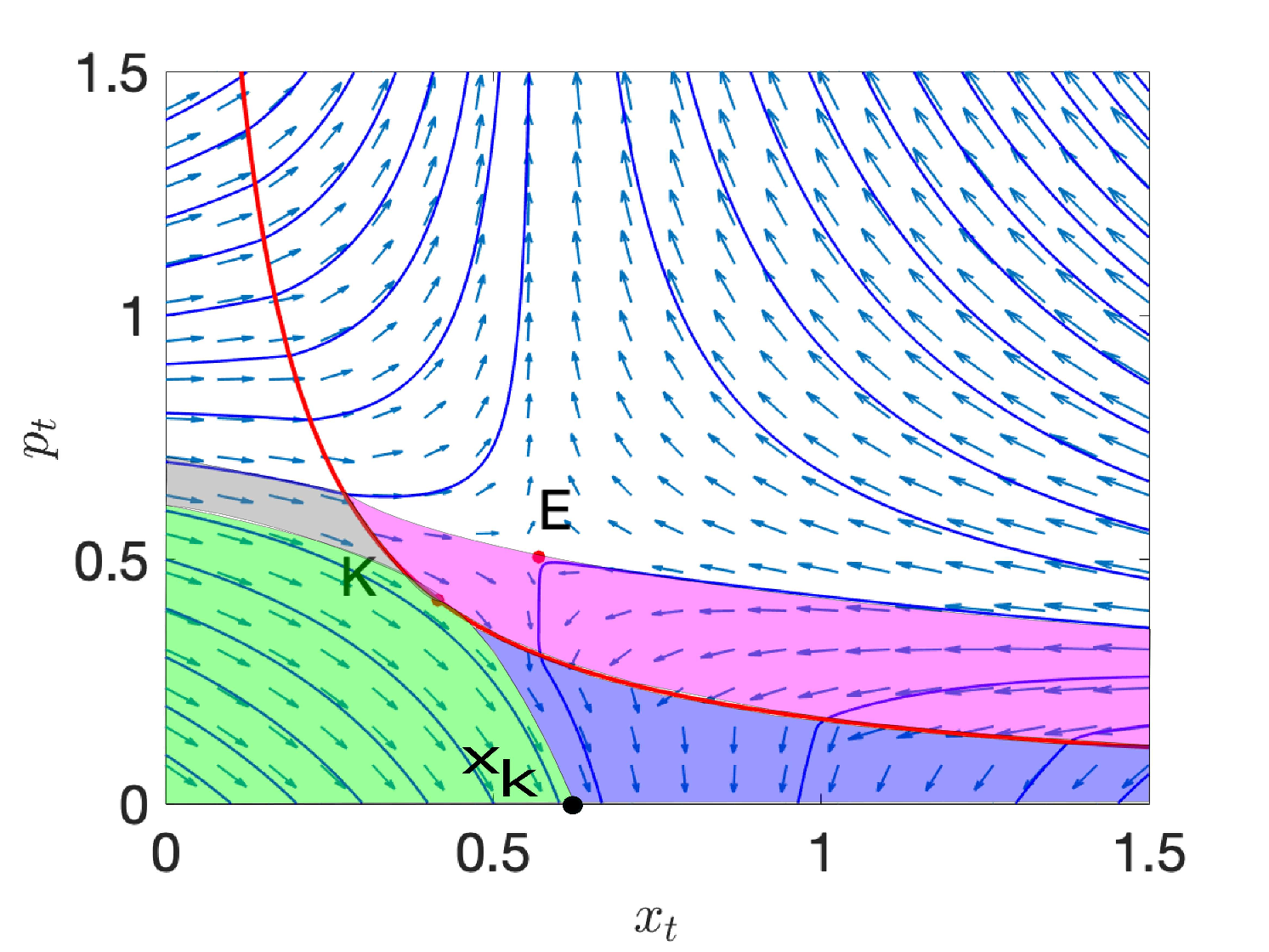}
\caption{Case C.}
\label{fig:case_c_color}
\end{subfigure}
 \caption{The colored regions depicting different subcases for every particular case.}
\label{fig:cases_color}
\end{figure*}

\subsubsection*{\underline{Subcase A-2} ($(x_{t_0}^\star,p_{t_0}^\star)$ is in the pink region)} In this case, there exists a switching time $t'\in (t_0, t_1)$. To compute $t'$, notice that a particular solution \eqref{eqn:region3_x} satisfying the initial condition $x_{t_0}=x_0$ is given by
\begin{equation}
      \hat{x}_t=\frac{1}{\gamma}\left(a+c-\frac{2c}{k_3e^{2 ct}+1}\right) \label{eqn:x_case2}
\end{equation}
where $c=\sqrt{a^2+\gamma}$ and
$k_3=\frac{c-a+x_0\gamma}{c+a-x_0\gamma}e^{-2ct_0}$.
On the other hand, the particular solution \eqref{eqn:region1_xp} satisfying $p_{t_1}=0$ is given by \eqref{eqn:p_case1}.
Thus, at the switching time $t'$, it must be that $\hat{x}_{t'} \bar{p}_{t'}=\alpha$, or
\begin{align}
\frac{1}{\gamma}\left(a+c-\frac{2c}{k_3e^{2 ct'}+1}\right)\left(\frac{e^{-2a(t'-t_1)}-1}{2a}\right)=\alpha.
\label{eqn:caseA_switching_time}
\end{align}
Therefore, $t'$ can be computed by solving a nonlinear equation \eqref{eqn:caseA_switching_time}. The procedure to compute the optimal control $u^\star_{[t_0, t_1]}$  is summarized in Algorithm~\ref{alg:case_a}.
\begin{algorithm}[H]
 \caption{Optimal solution for case A}
\hspace*{\algorithmicindent} \textbf{Input:} $a<0$, $\alpha$, $x_0$, $\gamma$, $t_0$ and $t_1$  \\
 \hspace*{\algorithmicindent} \textbf{Output:} $u_t^\star$
\begin{algorithmic}[1]
\State Compute $(\bar{x}_{t_0}, \bar{p}_{t_0})$ from \eqref{eqn:xp_case1};
\If {$\bar{x}_{t_0}\bar{p}_{t_0}\leq \alpha$}\Comment{Subcase A-1}
\State $u_t^\star\gets\{0_{[t_0,t_1]}\}$;
\Else\Comment{Subcase A-2}
\State $t'\gets\text{solve \eqref{eqn:caseA_switching_time} for $t'$}$;
\State $u_t^\star\gets\{\gamma_{[t_0,t']},0_{[t',t_1]}\}$;
\EndIf
 \end{algorithmic}
\label{alg:case_a}
\end{algorithm}

\subsubsection{Case B}
The initial coordinate $(x_{t_0}^\star,p_{t_0}^\star)$ of the optimal solution is in the colored region in
Fig.~\ref{fig:case_b_color}. Let $x_K$ be the $x$-coordinate at which a particular solution \eqref{eqn:region1_xp} to the vector field $f_1$ that passes through $K=E=(\sqrt{\alpha},\sqrt{\alpha})$ at a certain time $t'' (< t_1)$ reaches at the terminal time $t_1$. It is straightforward to show that $x_K=2a\alpha + 2\sqrt{\alpha}$. The time $t''$ of passing $K=E$ can also be computed as 
\begin{equation}
\label{eq:case_b_t''}
t_1 -t'' = \frac{1}{2a}\ln(2a\sqrt{\alpha}+1).
\end{equation}
\subsubsection*{\underline{Subcase B-1} ($(x_{t_0}^\star,p_{t_0}^\star)$ is in the green region or on the green curve)}
Notice that the solution $(\bar{x}_t, \bar{p}_t)$ to $f_1$ with boundary conditions $x_{t_0}=x_0$ and $p_{t_1}=0$ is still given by \eqref{eqn:xp_case1}.
If $\bar{x}_{t_1}\leq x_K$ occurs, the optimal solution is entirely in the green region and is characterized by \eqref{eqn:xp_case1}. There is no switching.
\subsubsection*{\underline{Subcase B-2} ($(x_{t_0}^\star,p_{t_0}^\star)$ is in the pink region)}
This case occurs when $\bar{x}_{t_1}> x_K$, $\bar{x}_{t_0}\bar{p}_{t_0}\leq \alpha$ and $x_0>\sqrt{\alpha}$. The optimal solution is entirely in the pink region and is characterized by \eqref{eqn:xp_case1}. Hence, there is no switching.
\subsubsection*{\underline{Subcase B-3} ($(x_{t_0}^\star,p_{t_0}^\star)$ is on the orange curve)}
This case occurs when \eqref{eqn:xp_case1} satisfies $\bar{x}_{t_1}> x_K$, $\bar{x}_{t_0}\bar{p}_{t_0}\leq \alpha$ and $x_0 \leq \sqrt{\alpha}$. In this case, the solution \eqref{eqn:xp_case1} over $t_0 \leq t \leq t_1$ is not contained in Region 1 and thus it is not a valid solution to the boundary problem of our interest. The optimal solution in this case is depicted as orange and green curves in Fig.~\ref{fig:case_b_color}. First, it follows the orange curve from $x_{t_0}=x_0$ to $x_{t'}=\sqrt{\alpha}$, stay on $E=K$ for $t'\leq t \leq t''$, and then follows the green trajectory from $x_{t''}=\sqrt{\alpha}$ to $x_{t_1}=x_K$.
From \eqref{eqn:x_case1}, the time $t'$ can be computed from
\begin{equation}
\label{eq:case_b_t'}
\sqrt{\alpha}=\frac{(2ax_0+1)e^{2a(t'-t_0)}-1}{2a}.
\end{equation}
There are two switches in the optimal control input: a switch from $u=0$ to $u=u^\star$ at time $t'$, and a switch from $u=u^\star$ to $u=0$ at time $t''$.
\subsubsection*{\underline{Subcase B-4} ($(x_{t_0}^\star,p_{t_0}^\star)$ is in the gray region)}
This case occurs when \eqref{eqn:xp_case1} satisfies $\bar{x}_{t_1}>x_K$ and $\bar{x}_{t_0}\bar{p}_{t_0}>\alpha$, and equation \eqref{eqn:caseA_switching_time} has a solution $t'$ in $[t_0, t_1]$. In this case, a single switching from the gray region to the pink region occurs at $t'$.
\subsubsection*{\underline{Subcase B-5} ($(x_{t_0}^\star,p_{t_0}^\star)$ is on the purple curve)}
This case occurs when \eqref{eqn:xp_case1} satisfies $\bar{x}_{t_1}>x_K$ and $\bar{x}_{t_0}\bar{p}_{t_0}>\alpha$, and equation \eqref{eqn:caseA_switching_time} does not have a solution $t'$ in $[t_0, t_1]$.
In this case, the optimal solution follows the black trajectory from $x_{t_0}=x_0$ to $x_{t'}=\sqrt{\alpha}$, stays at $E=K$ for $t'\leq t \leq t''$, and then follows the green trajectory from $x_{t''}=\sqrt{\alpha}$ to $x_{t_1}=x_K$.
From \eqref{eqn:region3_x}, the time $t'$ can be computed from
\begin{equation}
\label{eq:case_b5_t'}
\sqrt{\alpha}=\frac{1}{\gamma}\left(a+c-\frac{2c}{k_3e^{2 ct'}+1}\right)
\end{equation}
where $k_3=\frac{c-a+x_0\gamma}{c+a-x_0\gamma}e^{-2ct_0}$. The optimal control input switches twice: a switch from $u=\gamma$ to $u=u^\star$ at $t'$, and a switch from $u=u^\star$ to $u=0$ at $t''$.

\begin{algorithm}[H]
 \caption{Optimal solution for case B}
\hspace*{\algorithmicindent} \textbf{Input:} $a<0$, $\alpha$, $x_0$, $\gamma$, $t_0$ and $t_1$  \\
 \hspace*{\algorithmicindent} \textbf{Output:} $u_t^\star$
\begin{algorithmic}[1]
\State Compute $(\bar{x}_{t_0}, \bar{p}_{t_0})$ and $\bar{x}_{t_1}$ from \eqref{eqn:xp_case1};
\If {$\bar{x}_{t_1}\leq x_K$}\Comment{Subcase B-1}
    \State $u_t^\star\gets\{0_{[t_0,t_1]}\}$;
\ElsIf{$\bar{x}_{t_0}\bar{p}_{t_0}<\alpha$ and $x_0>\sqrt{\alpha}$}\Comment{Subcase B-2}
    \State $u_t^\star\gets\{0_{[t_0,t_1]}\}$;
\ElsIf{$\bar{x}_{t_0}\bar{p}_{t_0}<\alpha$ and $x_0\leq \sqrt{\alpha}$}\Comment{Subcase B-3}
    \State Compute $t'$ from \eqref{eq:case_b_t'};
    \State Compute $t''$ from \eqref{eq:case_b_t''};
    \State $u_t^\star\gets\{0_{[t_0,t']},u^\star_{[t',t'']},0_{[t'',t_1]}\}$;
\ElsIf{\eqref{eqn:caseA_switching_time} has a solution $t'$ in $[t_0, t_1]$}\Comment{Subcase B-4}   
    \State Compute $t'$ from \eqref{eqn:caseA_switching_time};
    \State $u_t^\star\gets\{\gamma_{[t_0,t']},0_{[t',t_1]}\}$;
\Else\Comment{Subcase B-5}
    \State Compute $t'$ from \eqref{eq:case_b5_t'};
    \State Compute $t''$ from \eqref{eq:case_b_t''};
    \State $u_t^\star\gets\{\gamma_{[t_0,t']},u^\star_{[t',t'']},0_{[t'',t_1]}\}$;
       \EndIf
 \end{algorithmic}
\label{alg:case_b}
\end{algorithm}




\subsubsection{Case C}
In this case, the initial state-costate pair $(x_{t_0}^\star,p_{t_0}^\star)$ of the optimal solution can belong to four different regions indicated by four different colors in Fig.~\ref{fig:case_c_color}. The boundaries of the pink region are defined by the switching surface $\mathcal{S}$ and the separatrices converging to the point $E$.
\subsubsection*{\underline{Subcase C-1} ($(x_{t_0}^\star,p_{t_0}^\star)$ is in the green region)} Let $x_K=2a\alpha+2\sqrt{\alpha}$ be the $x$-coordinate shown in Fig.~\ref{fig:case_c_color}.
Consider once again the trajectory \eqref{eqn:xp_case1} solving $f_1$ with the boundary conditions $x_{t_0}=x_0$ and $p_{t_1}=0$. If we have $\bar{x}_{t_1}\leq x_K$, then the optimal solution is entirely in the yellow region. No switching occurs in this case.
\subsubsection*{\underline{Subcase C-2} ($(x_{t_0}^\star,p_{t_0}^\star)$ is in the blue region)} Consider \eqref{eqn:xp_case1} again. If $\bar{x}_{t_1}> x_K$, $x_0 > \sqrt{\alpha}$ and $\bar{x}_{t_0}\bar{p}_{t_0}\leq \alpha$, then the trajectory \eqref{eqn:xp_case1} is entirely in the blue region. No switching occurs in this case.
\subsubsection*{\underline{Subcase C-3} ($(x_{t_0}^\star,p_{t_0}^\star)$ is in the pink region)}
In this case, a switching occurs once. Let $t'\in (t_0, t_1)$ be the switching time. Notice that the optimal solution follows the trajectory of the form
\begin{subequations}
\begin{align}
      &\hat{x}_t=\frac{1}{\gamma}\left(a+c-\frac{2c}{k_3e^{2 ct}+1}\right), \\
      & \hat{p}_{t}= k_4\frac{\left(k_3e^{2 ct}+1\right)^2}{e^{2 ct}}+ \frac{(1+\alpha\gamma)\left(k_3e^{2 ct}+1\right)}{2ck_3e^{2 ct}}
\end{align}
\label{eqn:x_hat}
\end{subequations}
for $t_0\leq t \leq t'$, and 
\begin{align}
   \bar{p}_t=\frac{e^{-2a(t-t_1)}-1}{2a} 
   \label{eqn:p_bar}
\end{align}
for $t'\leq t\leq t_1$. Thus, Subcase C-3 applies if the following set of equations in terms of unknowns $p_{t_0}, k_3, k_4$ and $t'$ admits a solution such that $t_0 < t' < t_1$ and $x_0 p_{t_0} > \alpha$:
\begin{subequations}
\label{eq:c3_condition}
\begin{align}
    &x_{t_0}^\star=x_0=\frac{1}{\gamma}\left(a+c-\frac{2c}{k_3e^{2ct_0}+1}\right) \label{eq:c3_condition1} \\
    & p_{t_0}^\star=k_4\frac{\left(k_3e^{2 ct_0}+1\right)^2}{e^{2 ct_0}}+ \frac{(1+\alpha\gamma)\left(k_3e^{2 ct_0}+1\right)}{2ck_3e^{2 ct_0}} \label{eq:c3_condition2} \\
    &\frac{1}{\gamma}\left(a+c-\frac{2c}{k_3e^{2 ct'}+1}\right)\left(\frac{e^{-2a(t'-t_1)}-1}{2a}\right)=\alpha \label{eq:c3_condition3} \\
    &k_4\frac{\left(k_3e^{2 ct'}+1\right)^2}{e^{2 ct'}}+ \frac{(1+\alpha\gamma)\left(k_3e^{2 ct'}+1\right)}{2ck_3e^{2 ct'}} \nonumber \\
    &\hspace{28ex}=\frac{e^{-2a(t'-t_1)}-1}{2a} \label{eq:c3_condition4}
\end{align}
\end{subequations}
Condition \eqref{eq:c3_condition3} ensures that $\hat{x}_{t'}\bar{p}_{t'}=\alpha$, and \eqref{eq:c3_condition4} ensures that $\hat{p}_{t'}=\bar{p}_{t'}$ (i.e., the transition from $\widehat{p}_t$ to $\bar{p}_t$ is continuous).
\subsubsection*{\underline{Subcase C-4} ($(x_{t_0}^\star,p_{t_0}^\star)$ is in the purple region)}
Switching occurs twice in this case. Let $t'$ and $t''$ be the first and the second switching times. The optimal solution follows the trajectory of the form
\[
\tilde{x}_t=\frac{k_1e^{2at}-1}{2a},\quad \tilde{p}_t=\frac{k_2e^{-2at}-1}{2a}
\]
for $t_0\leq t \leq t'$ and satisfies \eqref{eqn:x_hat}
for $t'\leq t \leq t''$. For $t''\leq t_1$, the $p$-coordinate of the optimal solution satisfies \eqref{eqn:p_bar}.
Therefore, Subcase C-4 applies if the following set of equations in terms of unknowns $k_1, k_2, k_3, k_4, p_{t_0}, t'$ and $t''$ admits a solution such that $x_0p_{t_0}\leq \alpha$ and $t_0 \leq t' < t'' < t_1$:
\begin{subequations}
\label{eq:c4_condition}
\begin{align}
    &x_{t_0}^\star=x_0=\frac{k_1e^{2at_0}-1}{2a} \label{eq:c4_condition1} \\
    &p_{t_0}^\star=\frac{k_2e^{-2at_0}-1}{2a} \label{eq:c4_condition2} \\ 
    &\left(\frac{k_1e^{2at'}-1}{2a}\right)\left(\frac{k_2e^{-2at'}-1}{2a}\right)=\alpha \label{eq:c4_condition3} \\
    &\frac{k_1e^{2at'}-1}{2a}=\frac{1}{\gamma}\left(a+c-\frac{2c}{k_3e^{2 ct'}+1}\right) \label{eq:c4_condition4} \\
    &\frac{k_2e^{-2at'}-1}{2a}=k_4\frac{\left(k_3e^{2 ct'}+1\right)^2}{e^{2 ct'}}+\nonumber\\
    &\quad\quad\quad\quad\quad\quad\quad\quad\quad\frac{(1+\alpha\gamma)\left(k_3e^{2 ct'}+1\right)}{2ck_3e^{2 ct'}} \label{eq:c4_condition5} \\
    &\frac{1}{\gamma}\left(a+c-\frac{2c}{k_3e^{2 ct''}+1}\right)\left(\frac{e^{-2a(t''-t_1)}-1}{2a}\right)=\alpha \label{eq:c4_condition6} \\
    &k_4\frac{\left(k_3e^{2 ct''}+1\right)^2}{e^{2 ct''}}+ \frac{(1+\alpha\gamma)\left(k_3e^{2 ct''}+1\right)}{2ck_3e^{2 ct''}} \nonumber \\
    &\hspace{28ex}=\frac{e^{-2a(t''-t_1)}-1}{2a} \label{eq:c4_condition7}
\end{align}
\end{subequations}
Conditions \eqref{eq:c4_condition3} and \eqref{eq:c4_condition6} ensure that $\tilde{x}_{t'}\tilde{p}_{t'}=\alpha$ and $\bar{x}_{t''}\hat{p}_{t''}=\alpha$ (i.e., switching happens on the switching surface). Conditions \eqref{eq:c4_condition4}, \eqref{eq:c4_condition5} and \eqref{eq:c4_condition7} ensure that the trajectory is continuous at switching times.
\begin{algorithm}[H]
 \caption{Optimal solution for case C}
\hspace*{\algorithmicindent} \textbf{Input:} $a<0$, $\alpha$, $x_0$, $\gamma$, $t_0$ and $t_1$  \\
 \hspace*{\algorithmicindent} \textbf{Output:} $u_t^\star$
\begin{algorithmic}[1]
\State Compute $(\bar{x}_{t_0}, \bar{p}_{t_0})$ and $\bar{x}_{t_1}$ from \eqref{eqn:xp_case1};
\If {$\bar{x}_{t_1}\leq x_K$}\Comment{Subcase C-1}
    \State $u_t^\star\gets\{0_{[t_0,t_1]}\}$;
\ElsIf{$\bar{x}_{t_0}\bar{p}_{t_0}<\alpha$ and $x_0>\sqrt{\alpha}$}\Comment{Subcase C-2}
    \State $u_t^\star\gets\{0_{[t_0,t_1]}\}$;
\ElsIf{\eqref{eq:c3_condition} has a solution}\Comment{Subcase C-3}
    \State Compute $t'$ from \eqref{eq:c3_condition};
    \State $u_t^\star\gets\{\gamma_{[t_0,t']}, 0_{[t',t_1]}\}$;
\Else\Comment{Subcase C-4}
    \State Compute $t'$ and $t''$ from \eqref{eq:c4_condition};
    \State $u_t^\star\gets\{0_{[t_0,t']},\gamma_{[t',t'']},0_{[t'',t_1]}\}$;
       \EndIf
 \end{algorithmic}
\label{alg:case_b}
\end{algorithm}

\section{Time-invariant solutions}
\label{sec:time-invariance}
We now turn our attention to the problem of finding the optimal time-invariant channel gain $C\in\mathbb{R}^{n\times n}$ as formulated in 
\eqref{problem:time_invariant}.
Since $(A, B)$ is controllable (we have assumed that $B\in\mathbb{R}^{n\times n}$ is nonsingular) and $(A, C)$ is detectable (we have assumed that $A$ is Hurwitz) for every $C\in\mathbb{R}^{n\times n}$, the algebraic Riccati equation
\begin{equation}
\label{eq:riccati}
AX+XA^\top-XC^\top CX+BB^\top=0
\end{equation} 
admits a unique positive semidefinite solution, which is positive definite  \cite[Theorem 13.7,
Corollary 13.8]{zhou1996robust}. In this case, it can also be shown that the
solution $X_t$ to the Riccati differential equation \eqref{problem:time_invariant_riccati_b}  
with the initial condition $X_{t_0}=X_0\succeq 0$ satisfies $X_t \rightarrow X$ as $t \rightarrow +\infty$
(e.g., \cite[Theorem 10.10]{bitmead1991riccati}), where $X$ is the
unique positive definite solution to \eqref{eq:riccati}. 
Therefore, it follows from the convergence of Ces\`{a}ro mean that
\begin{align*}
&\frac{1}{t_1-t_0}\int_{t_0}^{t_1} \text{Tr}(CX_tC^\top)dt \rightarrow \text{Tr}(CXC^\top) \\
&\frac{1}{t_1-t_0}\int_{t_0}^{t_1}  \text{Tr}(X_t)dt \rightarrow \text{Tr}(X) 
\end{align*}
as $t_1\rightarrow +\infty$. Thus, \eqref{problem:time_invariant_riccati} can be simplified as
\begin{subequations}
\label{problem:time_invariant_are}
\begin{align}
    \min_{C\in\mathbb{R}^{n\times n}, X\in\mathbb{S}_+^n}& \text{Tr}(X)+\alpha\text{Tr}(CXC^\top) \label{problem:time_invariant_are_a}\\
    \mathrm{s.t. }\qquad &AX+XA^\top-XC^\top CX+BB^\top=0 \label{problem:time_invariant_are_b} \\
    &C^\top C  \preceq \gamma I. \label{problem:time_invariant_are_c}
\end{align}
\end{subequations}
As the main result of this section, we show that \eqref{problem:time_invariant_are} can be reformulated as an equivalent semidefinite program with a rank constraint. The result is summarized in the next theorem.
\begin{theorem}
\label{thm:lmi}
Suppose $A\in\mathbb{R}^{n\times n}$ is Hurwitz and $B\in\mathbb{R}^{n\times n}$ is nonsingular. For any given positive constants $\alpha$ and $\gamma$, the following statements hold:
\begin{itemize}
\item[(i)] The optimal value of the problem  \eqref{problem:time_invariant} coincides with the value of the semidefinite program with a rank constraint:
\begin{subequations}
\label{eq:sdp_a}
\begin{align}
    \min_{X,  Y\in\mathbb{S}_+^n}& \text{Tr}(X)+\alpha\text{Tr}(B^\top YB)+2\alpha \text{Tr}(A) \label{eq:sdp_a1}\\
    \mathrm{s.t. } \quad &AX+XA^\top+BB^\top \succeq 0 \label{eq:sdp_a2} \\
    &\begin{bmatrix} YA+A^\top Y -\gamma I & YB \\ B^\top Y & -I\end{bmatrix} \preceq 0 \label{eq:sdp_a3} \\
        &\begin{bmatrix} X & I \\ I & Y \end{bmatrix} \succeq 0 \label{eq:sdp_a4} \\
         &\text{rank}\begin{bmatrix} X & I \\ I & Y \end{bmatrix} =n. \label{eq:sdp_a5} 
\end{align}
\end{subequations}
\item[(ii)] An optimal solution $(X^\star, Y^\star)$ to \eqref{eq:sdp_a} exists and satisfies $X^\star\succ 0$ and $Y^\star\succ 0$. Moreover, any matrix $C^\star \in\mathbb{R}^{n\times n}$ satisfying
\begin{equation}
\label{eq:c_reconstruct}
{C^\star}^\top C^\star=Y^\star A+A^\top Y^\star + Y^\star BB^\top Y^\star
\end{equation}
is an optimal solution to \eqref{problem:time_invariant}.
\item[(iii)] Suppose that a semidefinite program obtained by removing the rank constraint \eqref{eq:sdp_a5} from \eqref{eq:sdp_a} admits an optimal solution satisfying $X^\star \succ 0$, $Y^\star \succ 0$ and 
\[
\text{rank}\begin{bmatrix} X^\star  & I \\ I & Y^\star \end{bmatrix} =n.
\]
Then, any matrix $C^\star \in\mathbb{R}^{n\times n}$ satisfying \eqref{eq:c_reconstruct} is an optimal solution to \eqref{problem:time_invariant}.
\end{itemize}
\end{theorem}
\begin{remark}
In Section~\ref{sec:proof_lmi} below, we prove the equivalence between the original problem \eqref{problem:time_invariant} and the optimization problem \eqref{eq:sdp_a}.
Notice that off-the-shelf solvers are not applicable to \eqref{eq:sdp_a} because of the non-convex rank constraint.
Remarkably, however, in numerous numerical experiments we have performed (see Section~\ref{sec:experiments}), the convex relaxation obtained by dropping the rank constraint \eqref{eq:sdp_a5} always admitted a solution satisfying \eqref{eq:sdp_a5}.
Thus, as per statement (iii) of Theorem~\ref{thm:lmi}, an optimal solution $C^\star \in\mathbb{R}^{n\times n}$ to the original problem \eqref{problem:time_invariant} was always computable by solving the semidefinite program \eqref{eq:sdp_a1}-\eqref{eq:sdp_a4}.
Currently, it is not known to us whether the convex relaxation is always exact.
\end{remark}

\subsection{Proof of Theorem~\ref{thm:lmi}}
\label{sec:proof_lmi}
We will show the equivalence between \eqref{problem:time_invariant_are} and \eqref{eq:sdp_a}. First, notice that the constraint \eqref{problem:time_invariant_are_b} implies that $X$ is positive definite, and that
\begin{align*}
\text{Tr}(CXC^\top) &=\text{Tr}(XC^\top C XX^{-1}) \\
&=\text{Tr}(AX+XA^\top BB^\top) X^{-1} \\
&=2\text{Tr}(A)+\text{Tr}(B^\top X^{-1}B).
\end{align*}
Introducing $Y:=X^{-1}$, \eqref{problem:time_invariant_are} can be written as an equivalent problem:
\begin{subequations}
\label{eq:sdp_b}
\begin{align}
    \min \quad & \text{Tr}(X)+\alpha\text{Tr}(B^\top YB)+2\alpha \text{Tr}(A) \label{eq:sdp_b1}\\
    \mathrm{s.t. }\quad &AX+XA^\top-XC^\top CX+BB^\top = 0 \label{eq:sdp_b2} \\
    &YA+A^\top Y-C^\top C +YBB^\top Y = 0 \label{eq:sdp_b3} \\
    &X=Y^{-1} \label{eq:sdp_b4} \\
    &C^\top C \preceq \gamma I \label{eq:sdp_b5} 
\end{align}
\end{subequations}
with respect to the variables $X\in\mathbb{S}_+^n$, $Y\in\mathbb{S}_+^n$ and $C\in\mathbb{R}^{n\times n}$. Notice that \eqref{eq:sdp_b2} and \eqref{eq:sdp_b3} are redundant conditions under the constraint $X=Y^{-1}$.
Next, we claim that \eqref{eq:sdp_b} is equivalent to the following optimization problem: 
\begin{subequations}
\label{eq:sdp_c}
\begin{align}
    \min_{X, Y} \quad & \text{Tr}(X)+\alpha\text{Tr}(B^\top YB)+2\alpha \text{Tr}(A) \label{eq:sdp_c1}\\
    \mathrm{s.t. }\quad &AX+XA^\top+BB^\top \succeq 0 \label{eq:sdp_c2} \\
    &YA+A^\top Y-\gamma I +YBB^\top Y \preceq 0 \label{eq:sdp_c3} \\
    &X=Y^{-1} \label{eq:sdp_c4} \\
    &X\succ 0, \;\; Y\succ 0. \label{eq:sdp_c5} 
\end{align}
\end{subequations}
It is clear that the constraints \eqref{eq:sdp_b2}-\eqref{eq:sdp_b5} imply the constraints \eqref{eq:sdp_c2}-\eqref{eq:sdp_c5}. Conversely,  for any $(X, Y)$ satisfying \eqref{eq:sdp_c2}-\eqref{eq:sdp_c5}, it is always possible to construct a tuple $(X, Y, C)$ satisfying \eqref{eq:sdp_b2}-\eqref{eq:sdp_b5} by choosing $C\in\mathbb{R}^{n\times n}$ to satisfy
\[
C^\top C=YA+A^\top Y + YBB^\top Y.
\]
Applying the Schur complement formula to \eqref{eq:sdp_c3}, and noticing that $X=Y^{-1}\succ 0$ is equivalent to
\[
\begin{bmatrix} X & I \\ I & Y \end{bmatrix} \succeq 0  \text{ and } \text{rank}\begin{bmatrix} X & I \\ I & Y \end{bmatrix} =n,
\]
\eqref{eq:sdp_b} can be written as
\begin{subequations}
\label{eq:sdp_d}
\begin{align}
    \min_{X,  Y\in\mathbb{S}_+^n}& \text{Tr}(X)+\alpha\text{Tr}(B^\top YB)+2\alpha \text{Tr}(A) \label{eq:sdp_d1}\\
    \mathrm{s.t. } \quad &AX+XA^\top-XC^\top CX+BB^\top \succeq 0 \label{eq:sdp_d2} \\
    &\begin{bmatrix} YA+A^\top Y -\gamma I & YB \\ B^\top Y & -I\end{bmatrix} \preceq 0 \label{eq:sdp_d3} \\
        &\begin{bmatrix} X & I \\ I & Y \end{bmatrix} \succeq 0 \label{eq:sdp_d4} \\
         &\text{rank}\begin{bmatrix} X & I \\ I & Y \end{bmatrix} =n. \label{eq:sdp_d5}  \\
         & X \succ 0, \; Y \succ 0. \label{eq:sdp_d6}
\end{align}
\end{subequations}
It is left to show that \eqref{eq:sdp_a} has an optimal solution $(X^\star, Y^\star)$ and that $X^\star \succ 0, Y^\star \succ 0$.
The existence of an optimal solution is guaranteed by Weierstrass' theorem \cite[Proposition A.8]{bertsekas1995nonlinear}, since the feasible domain for $(X,Y)$ is closed and the objective function is coercive.
To show $X^\star \succ 0$ must be the case, consider a sequence of feasible points $(X_n, Y_n)$ such that $X_n \rightarrow X_0$ as $n \rightarrow \infty$ for some singular $X_0 \succeq 0$. Then, from the constraint \eqref{eq:sdp_d4}, it is necessary that the maximum singular value $\bar{\sigma}(Y_n)$ satisfies $\bar{\sigma}(Y_n)\rightarrow \infty$ as $n \rightarrow \infty$. 
This implies $\text{Tr}(B^\top Y_n B)\rightarrow \infty$ as $n\rightarrow \infty$ and thus the objective function \eqref{eq:sdp_d1} tends to infinity.
Thus, by continuity of the objective function, it cannot be that $X^\star=X_0$.
The necessity of $Y^\star \succ 0$ can be shown similarly.
Hence  \eqref{eq:sdp_d6} is a redundant condition and can be removed. Therefore, we have shown the equivalence between  \eqref{problem:time_invariant_are} and \eqref{eq:sdp_a}, establishing statement (i). Statements (ii) and (iii) also follow from the argument above.

\subsection{Numerical experiments}
\label{sec:experiments}
To demonstrate the effectiveness of the convex relaxation presented in Theorem~\ref{thm:lmi}, we analyzed the numerical solutions to the semidefinite program \eqref{eq:sdp_a1}-\eqref{eq:sdp_a4} for a large number of randomly generated matrix pairs $A\in\mathbb{R}^{n\times n}, B\in\mathbb{R}^{n\times n}$.
To generate a random stable state space model, we used Matlab's \texttt{rss} function. To solve \eqref{eq:sdp_a1}-\eqref{eq:sdp_a4}, we used a semidefinite programming solver SDPT3 \cite{tutuncu2003solving}. 

With $n=15$, $\alpha=0.01$ and $\gamma=100$, we generated 1,000 pairs of random matrices $(A, B)$. In each simulation, we observed that the solution to  \eqref{eq:sdp_a1}-\eqref{eq:sdp_a4} satisfied the rank condition \eqref{eq:sdp_a5} up to the precision of the SDP solver. Fig.~\ref{fig:rank} shows a typical outcome of the eigenvalues $\lambda_i$ of the matrix $\begin{bmatrix} X & I \\ I & Y \end{bmatrix} \in \mathbb{S}_+^{2n}$. It can be seen that the last $n$ eigenvalues are negligible compared to the first $n$.  
\begin{figure}[t]
\centering
\includegraphics[width=9.5cm]{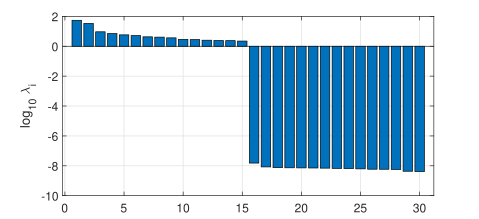}
\caption[Eigenvalues]{Eigenvalues $\lambda_i$, $i=1, 2, ... , 30$ of $\begin{bmatrix} X & I \\ I & Y \end{bmatrix}\in \mathbb{S}_+^{30}$ obtained as a numerical solution to the semidefinite program \eqref{eq:sdp_a1}-\eqref{eq:sdp_a4}.}
\label{fig:rank}
\end{figure}

In Fig.~\ref{fig:svd}, we plot a typical outcome of the eigenvalues of $C^\top C$ computed by \eqref{eq:c_reconstruct}. Again, we set $n=15$, and the upper limit of the channel gain is set to $\gamma=100$.
We solved \eqref{eq:sdp_a1}-\eqref{eq:sdp_a4} with different values of $\alpha$. As expected, the optimal channel gain tends to decrease as $\alpha$ increases. It is also noteworthy that both the lower saturation $\lambda_i(C^\top C)=0$ and the upper saturation $\lambda_i(C^\top C)=100$ can occur.

\begin{figure}[t]
\centering
\includegraphics[width=9.5cm]{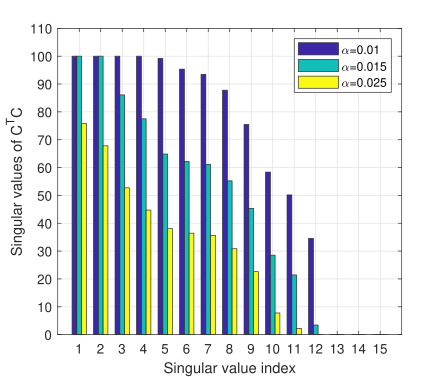}
\caption{Singular values of $C^\top C$ constructed from a numerical solution to the semidefinite program \eqref{eq:sdp_a1}-\eqref{eq:sdp_a4}.}
\label{fig:svd}
\end{figure}

\subsection{Scalar case}
We finally consider a special case of \eqref{problem:time_invariant} with $n=1$ to show that: (1) the exactness of the SDP relaxation discussed in Theorem~\ref{thm:lmi} (iii) holds in the scalar case, and (2) stationary points in the phase portrait we used to analyze the finite-horizon optimal channel gain control problem characterize the optimal time-invariant solution to the infinite-horizon channel gain control problem in the scalar case.
As in Section~\ref{sec:scalar}, we assume $A=a<0$, $B=1$, $\alpha>0$ and $\gamma>0$.

The next proposition shows that the rank condition \eqref{eq:sdp_a5} is automatically satisfied    even if it is not explicitly included in \eqref{eq:sdp_a} in the scalar case.
\begin{proposition}
The semidefinite program
\begin{subequations}
\label{eq:sdp_scalar}
\begin{align}
    \min_{x \geq0, y \geq 0} \quad & x+\alpha y + 2a\alpha  \label{eq:sdp_scalar1}\\
    \mathrm{s.t. }\quad & 2ax+1 \geq 0 \label{eq:sdp_scalar2} \\
    &\begin{bmatrix} 2ay-\gamma & y \\ y & -1 \end{bmatrix} \preceq 0, \; \begin{bmatrix} x & 1 \\ 1 & y \end{bmatrix} \succeq 0 \label{eq:sdp_scalar3} 
\end{align}
\end{subequations}
admits an optimal solution $(x^\star, y^\star)$. Moreover, $x^\star>0$, $y^\star>0$ and 
\begin{equation}
\label{eq:rank1}
    \text{rank}\begin{bmatrix} x^\star& 1 \\ 1 & y^\star \end{bmatrix}=1.
\end{equation}
\end{proposition}

\begin{proof}
The existence and the positivity of an optimal solution follow from the proof of Theorem~\ref{thm:lmi}. To show \eqref{eq:rank1}, it is sufficient to prove that $x^\star y^\star = 1$. To complete the proof by contradiction, suppose $(x^\star, y^\star)$ is an optimal solution but $x^\star >1/y^\star$. Set $x^{\star \star}:=1/y^\star$ and consider a new solution candidate $(x^{\star \star}, y^\star)$. Clearly, $(x^{\star \star}, y^\star)$ satisfies the constraints \eqref{eq:sdp_scalar3}. It also satisfies \eqref{eq:sdp_scalar2} since 
\[
2ax^{\star \star}+1 > 2ax^{\star}+1 \geq 0
\]
where we have used the fact that $a<0$ and $x^\star>x^{\star \star}$. Moreover, $x^\star>x^{\star \star}$ implies that $(x^{\star \star}, y^\star)$ attains a smaller objective value \eqref{eq:sdp_scalar1} than $(x^\star, y^\star)$ does. This contradicts the optimality of $(x^\star, y^\star)$.
\end{proof}

It is also fruitful to obtain an explicit solution to \eqref{problem:time_invariant_are} in the scalar case. Setting $u=C^2$,  \eqref{problem:time_invariant_are} can be written as
\begin{subequations}
\label{eq:ux_scalar}
\begin{align}
    \min_{u, x} \quad & x(1+\alpha u) \label{eq:ux_scalar1}\\
    \mathrm{s.t. }\quad & 2ax-ux^2+1=0, \; x\geq 0 \label{eq:ux_scalar2} \\
    &0\leq u\leq \gamma. \label{eq:ux_scalar3} 
\end{align}
\end{subequations}
Equation \eqref{eq:ux_scalar2} can be solved explicitly as $x=\frac{1}{\sqrt{a^2+u}-a}$, and the constraint \eqref{eq:ux_scalar3} can be written in terms of $x$ as $\frac{a+\sqrt{a^2+\gamma}}{\gamma}\leq x \leq -\frac{1}{2a}$. 
Expressing the objective function \eqref{eq:ux_scalar1} in $x$, the problem \eqref{eq:ux_scalar} is simplified to 
\begin{subequations}
\label{eq:x_scalar}
\begin{align}
    \min_{x} \quad & x+2a\alpha+\frac{\alpha}{x} \label{eq:x_scalar1}\\
    \mathrm{s.t. }\quad & \frac{a+\sqrt{a^2+\gamma}}{\gamma}\leq x \leq -\frac{1}{2a} \label{eq:x_scalar2}
\end{align}
\end{subequations}
which can be solved easily. From this analysis, it can be shown that the optimal solution $(x^\star, u^\star)$ to \eqref{eq:ux_scalar} is obtained as follows:
\begin{itemize}
    \item Case A: $1/4a^2<\alpha$. In this case, $(x^\star, u^\star)=(-\frac{1}{2a},0)$.
    \item Case B: $(a+\sqrt{a^2+\gamma})^2/\gamma^2\leq\alpha\leq 1/4a^2$. In this case, $(x^\star, u^\star)=\left(\sqrt{\alpha}, \frac{2a}{\sqrt{\alpha}}+\frac{1}{\alpha} \right)$.
    \item Case C: $\alpha<(a+\sqrt{a^2+\gamma})^2/\gamma^2$. In this case, $(x^\star, u^\star)=\left(\frac{a+\sqrt{a^2+\gamma}}{\gamma},\gamma \right)$.
\end{itemize}
Notably, these solutions coincide with the stationary points in the phase portrait we obtained in \eqref{eqn:equilibrium_region_1}-\eqref{eqn:equilibrium_region_3}. 

\section{Future work}
In this paper, we formulated a continuous-time optimal channel gain control problem for minimum-information Kalman-Bucy filtering. Our special focus has been on the optimal time-varying solution to finite-horizon problems for scalar processes and the optimal time-invariant solution to infinite-horizon problems for vector processes. 
The presented results can be extended to multiple directions in the future.
\begin{enumerate}
    \item As a natural generalization of Section~\ref{sec:scalar}, the optimal time-varying channel gain for multi-dimensional source processes should be investigated as future work.
    \item Although we considered the optimal \emph{time-invariant} solutions in Section~\ref{sec:time-invariance}, it remains to show that there always exists a time-invariant solution that minimizes the average cost \eqref{problem:time_invariant_a} over an infinitely long time horizon. The existence of a time-invariant optimal solution in discrete-time setting has been shown in \cite{tanaka2015semidefinite}.
    \item It remains to find a formal proof of, or a counterexample to disprove, the exactness of the SDP relaxation presented in Theorem~\ref{thm:lmi}.
    \item The problem formulation should be extended to controlled source processes. Such a problem has been considered in a discrete-time setting  \cite{tanaka2017lqg} where directed information has been used in place of mutual information. Directed information in continuous time has been introduced in \cite{weissman2012directed}.
    \item A coding-theoretic interpretation (operational meaning) of the problem studied in this paper needs to be clarified in the future.
\end{enumerate}

\bibliographystyle{IEEEtran}
\bibliography{main}
\end{document}